\documentclass[12pt]{article}
\usepackage{amsmath}
\usepackage{graphicx,psfrag,epsf}
\usepackage{enumerate}
\usepackage{natbib}
\usepackage{authblk}
\usepackage{bm}
\usepackage{multirow}
\usepackage{booktabs}
\usepackage{verbatim}
\usepackage{etoolbox}
\usepackage{amsthm}
\usepackage[noend]{algpseudocode}
\newtheorem{thm}{Theorem}[section]

\newcommand{\bcaption}[2]{\caption{\textbf{#1} #2}}
\newcommand{\beginsupplement}{%
        \setcounter{table}{0}
        \renewcommand{\thetable}{S\arabic{table}}%
        \setcounter{figure}{0}
        \renewcommand{\thefigure}{S\arabic{figure}}%
     }
\usepackage{algorithmicx,algorithm}
\newcommand{\code}[1]{\texttt{#1}}
\usepackage{url} % not crucial - just used below for the URL 

\newcommand{\blind}{0}
\newcommand*\samethanks[1][\value{footnote}]{\footnotemark[#1]}
\begin{document}

\def\spacingset#1{\renewcommand{\baselinestretch}%
{#1}\small\normalsize} \spacingset{1}

%%%%%%%%%%%%%%%%%%%%%%%%%%%%%%%%%%%%%%%%%%%%%%%%%%%%%%%%%%%%%%%%%%%%%%%%%%%%%%

\if0\blind
{
  \title{\bf A set of efficient methods to generate high-dimensional binary data with specified correlation structures}
  \author{Wei Jiang\thanks{
    These authors contributed equally to this work.}\hspace{.2cm}\\
    Department of Biostatistics, School of Public Health, Yale University, New Haven, CT, US\\
    and \\
    Shuang Song\samethanks
    \hspace{.2cm}\\
    Center for Statistical Science, Tsinghua University, Beijing 100084, China\\
    Department of Industrial Engineering, Tsinghua University, Beijing 100084, China\\
    and \\
    Lin Hou\hspace{.2cm}\\
  Center for Statistical Science, Tsinghua University, Beijing 100084, China\\
    Department of Industrial Engineering, Tsinghua University, Beijing 100084, China\\
    and \\
    Hongyu Zhao\thanks{
   Corresponding author:
               hongyu.zhao@yale.edu} \\
    Department of Biostatistics, School of Public Health, Yale University, New Haven, CT, US}
  \maketitle
} \fi

\if1\blind
{
  \bigskip
  \bigskip
  \bigskip
  \begin{center}
    {\LARGE\bf A set of efficient methods to generate high-dimensional binary data with specified correlation structures}
\end{center}
  \medskip
} \fi

\bigskip
\begin{abstract}
High dimensional correlated binary data arise  in many areas, such as observed genetic variations in biomedical research. 
 Data simulation can help researchers evaluate efficiency and explore properties of different computational and statistical  methods. Also, some statistical methods, such as Monte-Carlo methods, rely on data simulation. Lunn and Davies (1998) proposed linear time complexity methods to generate correlated binary variables with three common correlation structures. However, it is infeasible to specify unequal probabilities in their methods. In this manuscript, we introduce several computationally efficient algorithms that generate high-dimensional  binary data with specified correlation structures
 and unequal probabilities. Our algorithms have linear time complexity with respect to the dimension for three commonly studied correlation structures, namely exchangeable, decaying-product and $K$-dependent correlation structures. In addition, we extend our algorithms to generate binary data of general non-negative correlation matrices with quadratic time complexity.  We  provide an R package, CorBin, to implement our simulation methods.    Compared to the existing packages for binary data generation, the time cost to generate  a $100$-dimensional binary vector with the common correlation structures and general correlation matrices can be reduced up to $10^5$ folds  and $10^3$ folds, respectively, and the  efficiency can be further improved with the increase of dimensions. The R package CorBin is available on CRAN at https://cran.r-project.org/.
 \end{abstract}

\noindent%
\vfill

\newpage
\spacingset{1.45} % DON'T change the spacing!
\section{Introduction}
\label{sec:intro}
Binary data are commonly observed in many research areas, such as social survey responses, marketing data concerning specific issues with ``yes/no'' questions, responses to treatments in clinical trials, and measurements of genetic or epigenetic variations among  individuals. Usually, multiple variables are collected from an individual in these studies, and correlations among these variables are ubiquitous. Instead of considering each variable individually, it is essential to model the correlated variables together as a multi-dimensional vector \citep{cox1972analysis, carey1993modelling}.

With the improvements of technologies over the past decade, more and more variables are collected and constitute to high-dimensional data. One example is the  genomic data from biomedical research. Through high-throughput genotyping \citep{kennedy2003large} and sequencing technologies \citep{metzker2010sequencing}, millions of genetic variants of an individual can be collected simultaneously. It is well known that nearby genetic variants are often highly correlated, which is known as linkage disequilibrium \citep{pritchard2001linkage}. Hence we cannot simply regard each variant as an independent variable, and the variants constitute a high-dimensional correlated binary vector. Other examples of high-dimensional data include community DNA fingerprints data \citep{wilbur2002variable},
binary questionnaire data \citep{fieuws2006high},
consumer financial history data \citep{diwakar2009data} and consumer behavior data \citep{naik2008challenges}. 

Data simulation can help researchers evaluate the performance of a specific method when real data are difficult to access and/or there is no ground truth about the underlying model/mechanism. We can  explore the properties of the method through analyzing simulated data, such as exploring the small-sample properties of generalized estimating equations (GEE) \citep{hardin2002generalized}. In addition, many statistical estimation methods rely on the generation of random numbers, such as Monte-Carlo methods. With an ever increasing dimension in real world problems, it is essential to develop efficient simulation methods to generate high-dimensional data. In this article, we focus on how to efficiently generate high-dimensional binary data with specified marginal probabilities and correlation structures.

To generate random numbers, a direct method is to express the probability moment function and generate random samples with probabilities equal to the function values. 
Bahadur 
 established a parametric model expressing the joint mass function of binary variables with high-order correlations, but the distribution becomes computationally infeasible when the dimension is high \citep{bahadur1959representation}.
 During the 1990s, a number of more efficient methods were proposed to simulate multi-dimensional binary data. Emrich and Piedmonte  first identified a multivariate normal distribution whose pairwise copulas are linked with the specified correlation matrix \citep{emrich1991method}. The binary variables satisfying marginal and correlation conditions can be subsequently obtained by dichotomizing the identified distribution. This algorithm has been realized in an R package, bindata \citep{leisch1998generation}. One drawback of this method is that the multivariate normal distribution for a given correlation structure may not exist. In addition, a large number of non-linear equations with numerical integrations need to be solved in the algorithm. 
Lee proposed two algorithms based on linear programming to generate binary data with the exchangeable  correlation structure \citep{lee1993generating}. However, even with the structure assumption, a large number of non-linear equations are still needed.  
 Gange proposed to use iterative proportional fitting approach to generate multi-dimensional categorical variables, which is  more general than binary variables \citep{gange1995generating}. The extra iterative procedure makes the method computationally inefficient. This method has been implemented in the R package MultiOrd \citep{demirtas2006method}. 
 Similar with Emrich and Piedmonte's idea, Park et al. proposed to generate correlated binary data by dichotomizing correlated Poisson variables \citep{park1996simple}. The Poisson variables can be generated efficiently by summing independent Poisson variables, which may be shared among dimensions. Park et al.'s method is very efficient when the dimension is low. 
  Nevertheless, when the dimension grows high, it performs slower than the Emrich and Piedomonte's method.
  All the methods above are computationally inefficient when the dimension gets high.
  
 Instead of generating binary data with correlations specified by an arbitrary positive semi-definite matrix, Lunn and Davies proposed algorithms to generate binary data with three common correlation structures, including exchangeable, decaying-product and $1$-dependent \citep{lunn1998note}. The algorithms have linear time complexity with respect to the dimension.
 However, these algorithms only work when all the  marginal probabilities are equal, which limits the applicability of these algorithms. 
 
 Efficient simulating binary variables with unequal probabilities is critical in many applications. For example, let us assume we want to simulate single nucleotide polymorphism (SNP) data, most of which are biallelic \citep{sachidanandam2001map}, in order to study some statistical behavior of a certain genomic method under different correlation structures.  Since the allele frequencies (marginal probabilities) of different SNPs are naturally unequal, it will be unrealistic for us to model them with equal allele frequencies. Therefore, we may need to generate correlated binary data with a certain correlation structure and unequal marginal probabilities.
 The binary variables with unequal probabilities and specified correlation structures are also common in longitudinal study. A study on comparing two treatments for a common toenail infection is provided in Section 1.5.3 and Section 7.1 of \cite{shults2014quasi}, where 294 patients were randomized to one of two treatments, and measured at various time points (baseline, 1, 2, 3, 6, 9, and 12 months post baseline) to determine the presence or absence of a severe toenail infection. Hence, the outcome variables are dichotomous. They assumed an exchangeable correlation structure in each group of same treatment, and the probabilities of each time points are obviously unequal.
 In addition, this kind of data are commonly used in the GEE method, where the working correlation matrix are often specified as the forms we mentioned in the paper, and the probabilities are not necessarily equal.

In this paper, we generalize Lunn and Davies' algorithms  to generate high-dimensional correlated binary data with varied marginal probabilities. In line with their work, we first focus on three commonly used correlation structures including exchangeable, decaying-product, as well as $K$-dependent correlations.
Our algorithms have linear time complexity with respect to the dimension. Besides, we generalize the method on $K$-dependent structure and extend the applicability to general non-negative correlation matrices. Although the time complexity has been augmented to quadratic, the algorithm is still extremely efficient compared to existing approaches. We combine and implement these algorithms in an R package CorBin, which is publicly available on the Comprehensive R Archive Network (CRAN).
Compared with existing binary data generation packages, our package spends around $\frac{1}{700,000}$ of the time used in bindata and $\frac{1}{320,000}$ of MultiOrd when generating a 100-dimensional binary data with exchangeable structure. The ratios reach around $\frac{1}{4,500,000}$ and $\frac{1}{2,060,000}$ when the dimension increases to 500. Similar speed up is  also observed for the other two correlation structures. 
For generating data with general correlation matrices, the ratios become around $\frac{1}{680}$ and $\frac{1}{320}$ of bindata and MultiOrd, respectively.  
The package is easy to use, and a pdf document (\code{CorBin-manual.pdf}, Supplementary Material) is  provided to illustrate the usage for readers.

\begin{comment}

\begin{itemize}
\item Note that figures and tables (such as Figure~ and
Table~) should appear in the paper, not at the end or
in separate files.
\item In the latex source, near the top of the file the command
\verb+\newcommand{\blind}{1}+ can be used to hide the authors and
acknowledgements, producing the required blinded version.
\item Remember that in the blind version, you should not identify authors
indirectly in the text.  That is, don't say ``In Smith et. al.  (2009) we
showed that ...''.  Instead, say ``Smith et. al. (2009) showed that ...''.
\item These points are only intended to remind you of some requirements.
Please refer to the instructions for authors
at \url{http://amstat.tandfonline.com/action/authorSubmission?journalCode=utas20&page=instructions#.VFkojfnF_0c}
\item For more about ASA\ style, please see \url{http://journals.taylorandfrancis.com/amstat/asa-style-guide/}
\item If you have supplementary material (e.g., software, data, technical
proofs), identify them in the section below.  In early stages of the
submission process, you may be unsure what to include as supplementary
material.  Don't worry---this is something that can be worked out at later stages.
\end{itemize}
\end{comment}
\section{Models and algorithms}
\label{sec:meth}
In this section, we propose  algorithms for generating an $m$-dimensional random binary vector $\bm{X} = (X_1, X_2, \cdots, X_m)'$ with several correlation structures, where $X_i$ follows the Bernoulli distribution with marginal probability $p_i$, $i=1,2,\cdots,m.$ Compared with algorithms proposed in Lunn and Davies (1998), here we do not require all marginal probabilities to be equal, which extends the applicability  of the algorithms. In the following, we denote the correlation between $X_i$ and $X_j$ as  $r_{ij}$ ($1\leq i\leq j \leq n$). We assume  the correlations are non-negative, i.e. $r_{ij}\geq 0$, and use $\bm{R}$ to represent the correlation matrix constituted from $r_{ij}$. We also assumed all of the Bernoulli random variables used in each algorithm are generated independently.

\subsection{Natural restrictions  for correlated binary variables}\label{sec:restiction}
In this subsection, we present the natural restrictions  for any correlated binary variables must satisfy. First, the correlation matrix $\bm{R}$ should be positive definite. This imposes the restrictions on the correlation coefficients within $\bm{R}$ matrices. Taking  the $\bm{R}$  matrix presented in equation \eqref{1dep} (Section 2.4) as an example, $\rho_1$ must take values in $(-\frac{1}{c_m},\frac{1}{c_m})$ to satisfy the positive definite restriction, where $c_m=2sin(\frac{\pi(m-1)}{2(m+1)})$ \citep{shults2014quasi}. When $m \to +\infty$,  the constraints will approach $(-0.5,0.5)$.

Second, natural constraints of $\bm{R}$ are also imposed by the marginal expectations. Let's consider a simple bivariate example, in which $p_1, p_2=0.1, 0.4$, and $\rho=0.9$. The corresponding correlation matrix is obviously positive definite. However, $P(Y_1=1,Y_2=1)=0.1\times 0.4+0.9\sqrt{0.1\times 0.4\times (1-0.1)\times (1-0.4) }=0.172$; $P(Y_1=1,Y_2=0)=P(Y_1=1)-P(Y_1=1,Y_2=1)=0.1-0.172=-0.072<0$, which leads to an invalid probability. Prentice formulated the constraints imposed for marginal expectations and correlation coefficients, which are known as the Prentice constraints \citep{prentice1988correlated}. Specifically, under multivariate binary distribution, the correlation between any two dimensions $i, j$ should satisfy
\begin{equation}
    \max\left\{-\sqrt{\frac{p_i(1-p_j)}{p_j(1-p_i)}},
    -\sqrt{\frac{p_j(1-p_i)}{p_i(1-p_j)}}\right\}\leq r_{ij}\leq \min\left\{\sqrt{\frac{p_i(1-p_j)}{p_j(1-p_i)}},
    \sqrt{\frac{p_j(1-p_i)}{p_i(1-p_j)}}\right\}.
\end{equation}
For the  non-negative correlation structure, each element of the correlation matrix should satisfy $r_{ij}\in \left[0,\min\left\{\sqrt{\frac{p_i(1-p_j)}{p_j(1-p_i)}},
\sqrt{\frac{p_j(1-p_i)}{p_i(1-p_j)}}\right\}\right]$, based on the Prentice constraints. In our R package, we will first check whether the input parameters satisfy the Prentice constraints. If not, we will print out a warning message with exact values of the Prentice constraints for users.

Please note that the Prentice constraints and the positive definiteness are the necessary but insufficient conditions to guarantee the existence of the multivariate binary distributions with the specified correlation structures.
Actually, as we later presented in Section 2.2 and 2.3, for exchangeable and decaying-product correlation structure, the Prentice constraints are sufficient to guarantee the existence of the distributions. However, it is not true for all correlation structures. For example, 
\cite{chaganty2006range} presented a simple example with a 1-dependent
correlation structure (detailed in Section 2.4), in which the Prentice constraints were satisfied but the corresponding distribution is not valid.
Further, Shults and Hilbe provided a brief review on additional constraints for correlated binary data \citep{shults2014quasi}.

In following subsections, we provide implementation details and related properties of algorithms to generate binary data with different correlation structures.  

\subsection{Exchangeable correlation structure}\label{exchange}
The exchangeable correlation structure is one of the most commonly used structures in data simulation, and is regarded as the default setting in most binary data generation packages. In this case,  every pair of observations on a specific unit has the same correlation, i.e., 
\begin{equation}
r_{ij}=\left\{
\begin{aligned}
& \rho , \quad i \neq j \\
& 1 , \quad i = j\\
\end{aligned}
\right..
\end{equation}
The correlation matrix of the exchangeable correlation structures is:
\begin{equation}\label{corrmatrixEX}
\bm{R}={
\left( \begin{array}{ccccc}
1 & \rho & \rho &\cdots&\rho\\
\rho & 1 & \rho &\cdots&\rho\\
\rho & \rho & 1&\cdots&\rho\\
\vdots&\vdots&\vdots&\ddots&\vdots\\
\rho&\rho&\rho&\cdots&1
\end{array} 
\right )}.
\end{equation}
We use $p_{min}$ and $p_{max}$ to denote the minimal and maximal values in the desired marginal probabilities, i.e., 
\begin{equation}
    p_{min}=\min\limits_{k\in\{1,\cdots,m\}}{p_k},\quad
    p_{max}=\max\limits_{k\in\{1,\cdots,m\}}{p_k}.
\end{equation}

An intuitive thinking to generate binary data with exchangeable structure is to make each variable $X_i$ taking the linear combination form $(1-U_i)Y_i+U_iZ$, where $U_i\sim Bern(\alpha_i)$, $Y_i\sim Bern(\beta_i)$, $Z\sim Bern(\gamma)$, and all of the random variables are mutually independent. A careful selection of $\alpha_i$, $\beta_i$ and $\gamma$ is needed to make the constructed variables having specified marginal probabilities and correlations. We describe the constructions in detail in Algorithm \ref{alg1}. 

In the following, we first show the justification of the algorithm, i.e., if $\alpha_i$, $\beta_i$ and $\gamma$ lies in $\left[0,1\right]$, the binary data generated from Algorithm \ref{alg1} have specified marginal expectations and exchangeable correlation structure. After that, we prove that if the data to be generated satisfy the necessary condition of their existence, namely the Prentice constraints,   the construction of Algorithm \ref{alg1} can guarantee that the parameters lie in  $\left[0,1\right]$.

\begin{algorithm}[H]
\caption{Generate binary data under the exchangeable correlation structure.} 
\label{alg1}
\hspace*{0.02in} {\bf Input:}
The expected values of the Bernoulli random variables $p_1,p_2,\cdots,p_m$; and the correlation coefficient $\rho$ \\
\hspace*{0.02in} {\bf Output:}
The correlated binary variables $X_1,X_2,\cdots,X_m$
\begin{algorithmic}[1]
\makeatletter
\setcounter{ALG@line}{-1}
\makeatother
\State Check whether the input satisfies the Prentice constraints.
\State $\gamma=\frac{\sqrt{p_{min}p_{max}}}{\sqrt{p_{min}p_{max}}+\sqrt{(1-p_{min})(1-p_{max})}}$
\State Generate $Z\sim Bern(\gamma)$
\For{$i=1,2,\cdots,m$}
    \State $\alpha_i=\sqrt{\frac{\rho p_i(1-p_i)}{\gamma(1-\gamma)}}$
    \State $\beta_i=\frac{p_i-\alpha_i \gamma}{1-\alpha_i}$
    \State Generate $U_i\sim Bern(\alpha_i)$
    \State Generate $Y_i\sim Bern(\beta_i)$
    \State $X_i=(1-U_i)Y_i+U_iZ$
\EndFor
\State\Return $X_1,X_2,\cdots,X_m$
\end{algorithmic}
\end{algorithm}

%We first derive an intermediate variable $\gamma$:

%\begin{equation}
%    \gamma=\frac{\sqrt{p_{min}p_{max}}}{\sqrt{p_{min}p_{max}}+\sqrt{(1-p_{min})(1-p_{max})}}.
%\end{equation}
%Then for each $i\in\{1,2,\cdots,m\}$, let $\alpha_i=\sqrt{\frac{\rho p_i(1-p_i)}{\gamma(1-\gamma)}},\ \beta_i=\frac{p_i-\alpha_i \gamma}{1-\alpha_i}$.

% Let independent random variables $Z\sim Bern(\gamma),\ U_i\sim Bern(\alpha_i)$, and $\ Y_i\sim Bern(\beta_i),\ i\in \{1,2,\cdots,m\}$. Define $X_i=(1-U_i)Y_i+U_iZ$. Then the $(X_1,X_2,\cdots,m)'$ we generated has the above correlation structure. 

\begin{thm}\label{exTHM_cr} 
    If intermediate variables $\alpha_1$, $\cdots$, $\alpha_m$, $\beta_1$, $\cdots$, $\beta_m$, $\gamma \in \left[0,1\right]$, Algorithm \ref{alg1} returns the binary data  $X_1, \cdots,X_m$ with marginal probabilities $p_1,\cdots,p_m$ and common correlation $\rho$.
\end{thm} 
\begin{proof}
$X_i$ is a binary variable with the support \{0,1\}.
For any $i\in \{1,2,\cdots,m\}$, \begin{equation}\label{a1}
\begin{aligned}
  EX_i&=E((1-U_i)Y_i+U_iZ)=E(1-U_i)EY_i+EU_iEZ\\
  &=(1-\alpha_i)\beta_i+\alpha_i\gamma=(1-\alpha_i)\frac{p_i-\alpha_i \gamma}{1-\alpha_i}+\alpha_i \gamma=p_i,
\end{aligned}
\end{equation}

For any $i\neq j \in \{1,2,\cdots,m\}$
\begin{equation}\label{a2}
\begin{aligned}
r_{ij}&=\frac{cov(X_i,X_j)}{\sqrt{Var(X_i)Var(X_j)}} =\frac{cov((1-U_i)Y_i+U_iZ,(1-U_j)Y_j+U_jZ)}{\sqrt{p_ip_j(1-p_i)(1-p_j)}} \\
&=\frac{cov(U_iZ,U_jZ)}{\sqrt{p_ip_j(1-p_i)(1-p_j)}}
=\frac{\alpha_i\alpha_j\gamma(1-\gamma)}{\sqrt{p_ip_j(1-p_i)(1-p_j)}}=\rho.
\end{aligned}
\end{equation}

Thus, the algorithm returns the variables satisfying the required  marginal and correlation  conditions we provide.
\end{proof}

In the following theorem, we show that if the specified  binary data with non-negative exchangeable structure satisfies the Prentice constraints, the probability parameters $\alpha_i$, $\beta_i$ and $\gamma$ will automatically lie in the range of $\left[0,1\right]$.

\begin{thm}\label{exTHM} 
If the specified binary data with non-negative exchangeable correlation structure satisfy the Prentice constraints, the constructions in Algorithm \ref{alg1} guarantee that $\alpha_1$, $\cdots$, $\alpha_m$, $\beta_1$, $\cdots$, $\beta_m$, $\gamma$ lie in the range of  $\left[0, 1\right]$.
\end{thm} 
\begin{proof}
By definition, 
\begin{equation}
\label{eq:thm22_1}
    \gamma=\frac{\sqrt{p_{min}p_{max}}}{\sqrt{p_{min}p_{max}}+\sqrt{(1-p_{min})(1-p_{max})}}.
\end{equation}
Since $0 \leq \sqrt{p_{min}p_{max}} \leq \sqrt{p_{min}p_{max}}+\sqrt{(1-p_{min})(1-p_{max})}$, we have $0< \gamma< 1$.

Meanwhile, 
\begin{equation}\label{thm2.2_2}
    \gamma(1-\gamma)=\frac{\sqrt{p_{min}p_{max}(1-p_{min})(1-p_{max})}}{(\sqrt{p_{min}p_{max}}+\sqrt{(1-p_{min})(1-p_{max})})^2}.
\end{equation}
From  $Cauchy$-$Schwarz$ inequality, 
\begin{equation}\label{cs}
    (\sqrt{p_{min}p_{max}}+\sqrt{(1-p_{min})(1-p_{max})})^2\leq (p_{min}+1-p_{min})(p_{max}+1-p_{max})=1.
\end{equation}
According to \eqref{thm2.2_2} and \eqref{cs}, we have $\gamma(1-\gamma)\geq\sqrt{p_{min}p_{max}(1-p_{min})(1-p_{max})}.$

For any $i \in \{1,2,\cdots,m\}$, $\alpha_i$ is obviously non-negative.
With the Prentice constraints, we have 
\begin{equation}
    \label{pren.ex}
    \rho\leq \sqrt{\frac{p_{min}(1-p_{max})}{p_{max}(1-p_{min})}}.
\end{equation}
Therefore,
\begin{equation}
\begin{aligned}
  \rho p_i(1-p_i) &\leq
    \rho p_{max}(1-p_{min})\leq
    \sqrt{\frac{p_{min}(1-p_{max})}{p_{max}(1-p_{min})}}p_{max}(1-p_{min})\\
    &=\sqrt{p_{min}p_{max}(1-p_{min})(1-p_{max})}\leq
    \gamma(1-\gamma).
\end{aligned}
\end{equation}

Hence,  we have
\begin{equation}
\label{eq:thm22_2}
    \alpha_i=\sqrt{\frac{\rho p_i(1-p_i)}{\gamma(1-\gamma)}}\leq 1.
\end{equation}
%\For{$i=1,2,\cdots,m$}
 %   \State $\alpha_i=\sqrt{\frac{\rho p_i(1-p_i)}{\gamma(1-\gamma)}}$
  %  \State $\beta_i=\frac{p_i-\alpha_i \gamma}{1-\alpha_i}$
  
On the other hand, due to the fact that $p_{min}\leq p_i\leq p_{max}$,
\begin{equation}
\begin{aligned}
    &\frac{1-p_{max}}{p_{max}}\leq\frac{1-p_i}{p_i}
    \iff \sqrt{\frac{p_{min}(1-p_{max})}{p_{max}(1-p_{min})}}\leq \frac{1-p_i}{p_i}\frac{\sqrt{p_{min}p_{max}}}{\sqrt{(1-p_{min})(1-p_{max})}}=\frac{1-p_i}{p_i}\frac{\gamma}{1-\gamma},\\
    \text{and}\\
    &\frac{p_{min}}{1-p_{min}}\leq\frac{p_i}{1-p_i}
    \iff
    \sqrt{\frac{p_{min}(1-p_{max})}{p_{max}(1-p_{min})}}\leq \frac{p_i}{1-p_i}\frac{\sqrt{(1-p_{min})(1-p_{max})}}{\sqrt{p_{min}p_{max}}}=\frac{p_i}{1-p_i}\frac{1-\gamma}{\gamma},
\end{aligned}
\end{equation}
i.e.,
\begin{equation}
    \sqrt{\frac{p_{min}(1-p_{max})}{p_{max}(1-p_{min})}}\leq
    \min\{\frac{1-p_i}{p_i}\frac{\gamma}{1-\gamma},\frac{p_i}{1-p_i}\frac{1-\gamma}{\gamma}\}.
\end{equation}
Combined with inequality \eqref{pren.ex}, 
\begin{equation}\label{eq:thmpr1}
\begin{aligned}
  \rho\leq \sqrt{\frac{p_{min}(1-p_{max})}{p_{max}(1-p_{min})}}\leq \min\{\frac{1-p_i}{p_i}\frac{\gamma}{1-\gamma},\frac{p_i}{1-p_i}\frac{1-\gamma}{\gamma}\}.
\end{aligned}
\end{equation}
Then we can derive that
\begin{equation}\label{eq:thmpr2}
\begin{aligned}
 \rho\leq\frac{1-p_i}{p_i}\frac{\gamma}{1-\gamma}
\iff&
    \frac{\rho p_i(1-p_i)}{\gamma(1-\gamma)}\leq \frac{(1-p_i)^2}{(1-\gamma)^2}
  \iff
  \alpha_i=\sqrt{\frac{\rho p_i(1-p_i)}{\gamma(1-\gamma)}}\leq\frac{1-p_i}{1-\gamma}\\
  \iff&
  p_i-\alpha_i\gamma\leq 1-\alpha_i
  \iff
\beta_i=\frac{p_i-\alpha_i\gamma}{1-\alpha_i}\leq1.
\end{aligned}
\end{equation}
  
Similarly, 
\begin{equation}\label{eq:thmpr3}
\begin{aligned}
 \rho\leq\frac{p_i}{1-p_i}\frac{1-\gamma}{\gamma}
     \iff&
         \frac{\rho p_i(1-p_i)}{\gamma(1-\gamma)}\leq\frac{p_i^2}{\gamma^2}
         \iff
             \alpha_i=\sqrt{\frac{\rho p_i(1-p_i)}{\gamma(1-\gamma)}}\leq\frac{p_i}{\gamma}\\
             \iff&    
                 p_i-\alpha_i\gamma\geq 0 
                 \iff
    \beta_i\geq 0.
\end{aligned}
\end{equation} 
Combining \eqref{eq:thm22_1}, \eqref{eq:thm22_2}, \eqref{eq:thmpr2} and \eqref{eq:thmpr3}, we finish the proof.  
\end{proof}

\subsection{Decaying-product correlation structure}\label{auto}

The commonly used first order autoregressive (AR(1)) correlation structure is a special case of the decaying-product correlation structure, where the correlations are highest for adjacent variables and decrease in the power of  the distance between dimension indices. The correlation matrix of AR(1) is as follows: 

\begin{equation}
\bm{R}={
\left( \begin{array}{ccccc}
1 & \rho & \rho^2 &\cdots&\rho^{n-1}\\
\rho & 1 & \rho &\cdots&\rho^{n-2}\\
\rho^2 & \rho & 1&\cdots&\rho^{n-3}\\
\vdots&\vdots&\vdots&\ddots&\vdots\\
\rho^{n-1}&\rho^{n-2}&\rho^{n-3}&\cdots&1
\end{array} 
\right )}.
\end{equation}

Here we consider the more general decaying-product correlation structure. We allow  the marginal probabilities to be freely specified. Given the elements on the minor diagonal of the correlation matrix $\bm{\rho}=(\rho_1,\cdots,\rho_{n-1})$, the correlation between any two variables under this structure can be expressed as:
\begin{equation}
    r_{j,k}=\prod_{l=j}^{k-1}\rho_{l},\quad j<k,
\end{equation}

i.e.,
\begin{equation}
\bm{R}={
\left( \begin{array}{cccccc}
1 & \rho_1 & \rho_1\rho_2 &\rho_1\rho_2\rho_3 &\cdots&\prod_{l=1}^{n-1}\rho_{l}\\
\rho_1 & 1 & \rho_2 &\rho_2\rho_3 &\cdots&\prod_{l=2}^{n-1}\rho_{l}\\
\rho_1\rho_2 & \rho_2 & 1&\rho_3&\cdots&\prod_{l=3}^{n-1}\rho_{l}\\
\rho_1\rho_2\rho_3 & \rho_2\rho_3 &\rho_3& 1&\cdots&\prod_{l=4}^{n-1}\rho_{l}\\
\vdots&\vdots&\vdots&\vdots&\ddots&\vdots\\
\prod_{l=1}^{n-1}\rho_{l}&\prod_{l=2}^{n-1}\rho_{l}&\prod_{l=3}^{n-1}\rho_{l}&\prod_{l=4}^{n-1}\rho_{l}&\cdots&1
\end{array} 
\right )}.
\end{equation}

Similar with the construction method in the previous subsection, we also assume the variable, $X_i$ takes the linear combination form $(1-U_i)Y_i+U_iX_{i-1}$, where $U_i\sim Bern(\alpha_i)$, $Y_i\sim Bern(\beta_i)$, and all of the random variables are mutually independent. Here the construction of $X_i$ uses the information of $X_{i-1}$,  bringing the correlation between adjacent elements. Similar to the exchangeable correlation structure, we make a subtle construction of $\alpha_i$ and $\beta_i$ in order that the constructed variables have specified marginal probability and correlations.
We describe the constructions in detail in Algorithm \ref{alg2}. In the following, we  will first show the justification of the algorithm in Theorem \ref{arTHM_cr} and prove that the Prentice constraints are enough to guarantee the intermediate parameters in the interval $\left[0,1\right]$ in Theorem \ref{arTHM}.

\begin{algorithm}[H]
\caption{Generate binary data under decaying-product correlation structure.} 
\label{alg2}
\hspace*{0.02in} {\bf Input:}
The expected values of the Bernoulli random variables $p_1,p_2,\cdots,p_m$; and the off-diagonal correlation vector $\bm{\rho}$ \\
\hspace*{0.02in} {\bf Output:}
The correlated binary variables $X_1,X_2,\cdots,X_m$
\begin{algorithmic}[1]
    \makeatletter
\setcounter{ALG@line}{-1}
\makeatother
\State Check whether the input satisfies the Prentice constraints.
\State Generate $X_1\sim Bern(p_1)$
\For{$i=2,\cdots,m$}
    \State $\alpha_i=\rho_{i-1}\sqrt{\frac{p_i(1-p_i)}{p_{i-1}(1-p_{i-1})}}$
    \State $\beta_i=\frac{p_i-\alpha_i p_{i-1}}{1-\alpha_i}$
    \State
    Generate $U_i\sim Bern(\alpha_i)$
    \State
    Generate $Y_i\sim Bern(\beta_i)$
    \State
    $X_i = (1-U_i)Y_i+U_iX_{i-1}$
\EndFor
\State\Return $X_1,X_2,\cdots,X_m$
\end{algorithmic}
\end{algorithm}

\begin{thm}\label{arTHM_cr} 
    If intermediate variables $\alpha_2$, $\cdots$, $\alpha_m$, $\beta_2$, $\cdots$, $\beta_m \in \left[0,1\right]$, Algorithm \ref{alg2} returns the binary data with the specified marginal probabilities and decaying-product  correlation.
\end{thm} 
\begin{proof}
In  Algorithm \ref{alg2}, $EX_1=p_1$ is naturally satisfied. Now we prove $EX_i=p_i$ ($i=2,\cdots,m$) by induction. For any $i \in \{2,\cdots,m \}$, assuming that $EX_{i-1}=p_{i-1}$, then
\begin{equation}
    EX_i=E((1-U_i)Y_i+U_iX_{i-1})=(1-\alpha_i)\beta_i +\alpha_iEX_{i-1}
    =p_i-\alpha_ip_{i-1}+\alpha_ip_{i-1}=p_i.
\end{equation}
For any $k \in \{1,2,\cdots,m-1\}$ and $i \in \{1,2,\cdots,m-k\}$,
\begin{equation}
\begin{aligned}
  r_{i,i+k}&=\frac{cov(X_i,X_{i+k})}{\sqrt{Var(X_i)Var(X_{i+k})}}
  =\frac{cov(X_i,U_{i+k}U_{i+k-1}\cdots U_{i+1}X_{i})}{\sqrt{p_ip_{i+k}(1-p_i)(1-p_{i+k})}}\\
  &=\alpha_{i+k}\cdots\alpha_{i+1}\sqrt{\frac{p_i(1-p_i)}{p_{i+k}(1-p_{i+k})}}\\
  &=\rho_{i}\rho_{i+1}\cdots\rho_{i+k-1}.
\end{aligned}
\end{equation}
Thus the decaying-product correlation structure holds.
\end{proof}

\begin{thm}\label{arTHM}
    If the specified binary data with non-negative decaying-product correlation structure satisfy the Prentice constraints, the constructions in Algorithm \ref{alg2} guarantee that $\alpha_2$, $\cdots$, $\alpha_m$, $\beta_2$, $\cdots$, $\beta_m$ lie in the range of $\left[0,1\right]$.
\end{thm} 
\begin{proof}
For each $i \in \{2,\cdots,m\}$, we show $0\leq \alpha_i,\beta_i\leq 1$ in  two complementary cases:

When $p_{i-1}\leq p_i$, according to the Prentice constraints,  we have
\begin{equation}
    0\leq \alpha_i=\rho_{i-1}\sqrt{\frac{p_i(1-p_i)}{p_{i-1}(1-p_{i-1})}}\leq
    \sqrt{\frac{p_i(1-p_i)}{p_{i-1}(1-p_{i-1})}}\sqrt{\frac{p_{i-1}(1-p_{i})}{p_{i}(1-p_{i-1})}}=\frac{1-p_i}{1-p_{i-1}}\leq1.
\end{equation}
Thus,
\begin{equation}
\begin{aligned}
(1-p_{i-1})\alpha_i\leq1-p_i\iff
p_i-\alpha_i p_{i-1}\leq1-\alpha_i\iff
 \beta_i=\frac{p_i-\alpha_i p_{i-1}}{1-\alpha_i}\leq1.
 \end{aligned}
\end{equation}

At the same time, 
\begin{equation}
    \beta_i=\frac{p_i-\alpha_i p_{i-1}}{1-\alpha_i}\geq\frac{p_i-\alpha_i p_i}{1-\alpha_i}=p_i\geq0.
\end{equation}
Then we have $0\leq\beta_i\leq1$.

In the other case, i.e., $p_{i-1}> p_i$, according to the Prentice constraints, 

\begin{equation}
    0\leq \alpha_i=\rho_{i-1}\sqrt{\frac{p_i(1-p_i)}{p_{i-1}(1-p_{i-1})}}\leq
    \sqrt{\frac{p_i(1-p_i)}{p_{i-1}(1-p_{i-1})}}\sqrt{\frac{p_{i}(1-p_{i-1})}{p_{i-1}(1-p_{i})}}=\frac{p_i}{p_{i-1}}\leq1.
\end{equation}
Then we have 
\begin{equation}
\begin{aligned}
 p_i\geq\alpha_i p_{i-1}\iff
       \beta_i=\frac{p_i-\alpha_i p_{i-1}}{1-\alpha_i}\geq 0.
    \end{aligned}
\end{equation}

Meanwhile,
\begin{equation}
    \beta_i=\frac{p_i-\alpha_i p_{i-1}}{1-\alpha_i}\leq\frac{p_i-\alpha_i p_{i}}{1-\alpha_i}=p_i\leq1.
\end{equation}
At this point, we have proved that, with the Prentice constraints satisfied, $0\leq \alpha_i,\beta_i\leq 1$, $(i = 2,\cdots,m)$ hold for any binary data with non-negative decaying-product correlation structure, indicating they can be generated by Algorithm \ref{alg2}. 
\end{proof}

\subsection{1-dependent correlation structure}\label{statm}
Under the stationary $K$-dependent structure, there is a band of stationary correlations, such that each of the correlation is truncated to zero  after the $K$-th order band \citep{hardin2002generalized}. The correlation coefficient is:
\begin{equation}
r_{ij}=\left\{
\begin{aligned}
&\rho_{\left|i-j\right|},& \quad &\text{if} \  i\neq j \ \text{and}\  \left|i-j\right|\leq K \\
&0,& \quad &\text{if}\  \left|i-j\right|> K\\
& 1 ,& \quad &\text{if}\  i = j\\
\end{aligned}
\right..
\end{equation}

As the most common case, the stationary  $1$-dependent correlation matrix can be expressed as:
\begin{equation}
    \label{1dep}
\bm{R}={
\left( \begin{array}{cccccc}
1 & \rho_{1} & 0 &\cdots&0&0\\
\rho_{1} & 1 & \rho_{1} &\cdots&0& 0\\
0 & \rho_{1} & 1&\cdots&0& 0\\
\vdots&\vdots&\vdots&\ddots&\vdots&\vdots\\
0&0&0&\cdots&1&\rho_{1}\\
0 & 0& 0 &\cdots&\rho_{1}&1
\end{array} 
\right )},
\end{equation}
where the elements on the minor diagonal are equal. To expand the applicability of our method, we design the algorithms on a more general 1-dependent case, allowing the elements on the minor diagonal to vary, i.e., $\bm{\rho_1}=(\rho_{11},\rho_{12},\cdots,\rho_{1,m-1})$. For simplicity, we use $\bm{\rho}$ to represent $\bm{\rho_1}$ in this section.

In the main context, we focus on the generation of binary data with  $1$-dependent correlation structure. We introduce two algorithms with different applicable conditions. Both algorithms allow marginal probabilities to vary. We will describe the details of applicable conditions separately after introducing the corresponding algorithms. 

Intuitively, we intend to construct the variables that have correlation between adjacent elements, but independent with those not adjacent, in order to satisfy the $1$-dependent correlation structure. Our Algorithm \ref{alg3} utilizes the form of $X_i=U_iY_iY_{i-1}$ to guarantee the independence of $X_j$ and $X_k$ when $|j-k|\geq 2$. In addition, the intermediate parameters are also important to make sure that the correlation coefficients, and the marginal probabilities are satisfied. We present the details in Algorithm \ref{alg3}.

\begin{algorithm}[H]
\caption{Generate binary data under the  $1$-dependent correlation structure.  (Method \uppercase\expandafter{\romannumeral1})} 
\label{alg3}
\hspace*{0.02in} {\bf Input:}
The expected values of the Bernoulli random variables $p_1,p_2,\cdots,p_m$; and the correlation coefficient vector $\bm{\rho}=(\rho_1,\rho_2,\cdots,\rho_{m-1})$ \\
\hspace*{0.02in} {\bf Output:}
The correlated binary variables $X_1,X_2,\cdots,X_m$
\begin{algorithmic}[1]
    \makeatletter
    \setcounter{ALG@line}{-1}
    \makeatother
    \State Check whether the input satisfies the Prentice constraints.
\State $\beta_0=1$
\State $Y_0 \sim Bern(1)$
\For{$i=1,2,\cdots,m-1$}
    \State $\beta_i=\frac{\sqrt{p_ip_{i+1}}}{\sqrt{p_ip_{i+1}}+\rho_i\sqrt{(1-p_i)(1-p_{i+1})}}$
    \State $\alpha_i=\frac{p_i}{\beta_i\beta_{i-1}}$
    \State \textbf{if} $\alpha_i \textgreater 1$: Prompt the unavailability; \textbf{break}
    \State
    Generate $U_i\sim Bern(\alpha_i)$
    \State
    Generate $Y_i\sim Bern(\beta_i)$
    \State
    $X_i = U_iY_iY_{i-1}$
\EndFor
\State $\alpha_m=\sqrt{\frac{p_m}{\beta_{m-1}}}$
\State $\beta_m=\sqrt{\frac{p_m}{\beta_{m-1}}}$
\State Generate $U_m\sim Bern(\alpha_m)$
\State Generate $Y_m\sim Bern(\beta_m)$
\State $X_m=U_mY_mY_{m-1}$
\State\Return $X_1,X_2,\cdots,X_m$

\end{algorithmic}
\end{algorithm}

\begin{thm}
\label{kdepTHM_cr}
If intermediate variables $\alpha_1$, $\cdots$, $\alpha_m$, $\beta_1$, $\cdots$, $\beta_m \in \left[0,1\right]$, Algorithm \ref{alg3}  returns the  binary data with given  marginal expectation and  1-dependent correlation structure.
\end{thm}
\begin{proof}
For each $i \in \{1,2,\cdots,m-1\}$, 
\begin{equation}
    EX_i=EU_iY_iY_{i-1}=\alpha_i\beta_i\beta_{i-1}=\frac{p_i}{\beta_i\beta_{i-1}}\beta_i\beta_{i-1}=p_i.
\end{equation}
Meanwhile, 
\begin{equation}
    EX_m=EU_mY_mY_{m-1}=\alpha_m\beta_m\beta_{m-1}=\sqrt{\frac{p_m}{\beta_{m-1}}}\sqrt{\frac{p_m}{\beta_{m-1}}}\beta_{m-1}=p_m.
\end{equation}
From the generation process, it is obvious that $r_{ij}=0$ when $ \left|j-i\right|>1$.
Then for $i \in \{1,2,\cdots, m-1\}$,
\begin{equation}
\begin{aligned}
    r_{i,i+1}=&\frac{cov(X_i,X_{i+1})}{\sqrt{Var(X_i)Var(X_{i+1})}}=\frac{cov(U_iY_iY_{i-1},U_{i+1}Y_{i+1}Y_i)}{\sqrt{p_ip_{i+1}(1-p_i)(1-p_{i+1})}}\\
    =&\frac{E(U_iU_{i+1}Y_{i-1}Y_i^2Y_{i+1})-E(U_iY_iY_{i-1})E(U_{i+1}Y_{i+1}Y_i)}{\sqrt{p_ip_{i+1}(1-p_i)(1-p_{i+1})}}\\
   % =&\frac{\alpha_{i}\alpha_{i+1}\beta_{i-1}\beta_i\beta_{i+1}(1-\beta_i)}{\sqrt{p_ip_{i+1}(1-p_i)(1-p_{i+1})}}
    =&\frac{(\alpha_{i}\beta_{i-1}\beta_i)(\alpha_{i+1}\beta_i\beta_{i+1})(1-\beta_i)}{\beta_i\sqrt{p_ip_{i+1}(1-p_i)(1-p_{i+1})}}
    \\=&\frac{(1-\beta_i)\sqrt{p_ip_{i+1}}}{\beta_i\sqrt{(1-p_i)(1-p_{i+1})}}=\rho_i.
\end{aligned}
\end{equation}
\end{proof}

Similar to the previous algorithms, the data generated by Algorithm \ref{alg3} require intermediate parameters $\alpha_i$ and $\beta_i$ to be in the range of $[0,1]$. Otherwise, intermediate variables $U_i$ and $Y_i$ cannot be generated, nor can $X_i$. In the following theorem, we show the restriction for $\pmb{\rho}$ such that the intermediate parameter requirement is satisfied.

%For simplicity, we set the marginal distribution of each variable to be identical.

%% Algorithm3

\begin{thm}
\label{alg3_property}
Only binary data with non-negative 1-dependent structure satisfying the following inequalities can be generated from Algorithm \ref{alg3}:
\begin{equation}
    A_i\rho_{i-1}\rho_i+B_{1i}\rho_{i-1}+B_{2i}\rho_i\leq C_i,\quad i=2,\cdots,m-1,
    \label{alg_inequal}
\end{equation}
    where $A_i=\sqrt{(1-p_{i-1})(1-p_i)(1-p_{i+1})}$, 
$B_{1i}=\sqrt{(1-p_{i-1})p_ip_{i+1}}$, $B_{2i}=\sqrt{p_{i-1}p_i(1-p_{i+1})}$, and
$C_i=\sqrt{p_{i-1}(1-p_i)p_{i+1}}$. When $\rho_1=\rho_2=\dots=\rho_m=\rho$, we denote $A=A_i$, $B=B_{1i}+B_{2i}$ and $C=C_i$, then the inequality \eqref{alg_inequal} can be simplified to 
\begin{equation}
    0\leq\rho_i\leq \frac{-B+\sqrt{B^2+4AC}}{2A}.
\end{equation}
Further, in the special case when $p_1=p_2=\dots=p_m=p$, the applicable condition is
$0\leq\rho\leq\frac{\sqrt{p}}{1+\sqrt{p}}$.
\end{thm}
\begin{proof}
Since $\rho_i\geq0$, it is obvious that $\alpha_i\geq 0$ and $0\leq \beta_i\leq 1$ based on the definitions of $\alpha_i$ and $\beta_i$. Now we study the applicable condition for $\rho_i$ such that the intermediate parameter $\alpha_i\leq 1$. 
For each $i\in \{1,2,\cdots,m-1\}$, we have
\begin{equation}
\begin{aligned}
&\alpha_i=\frac{p_i}{\beta_i\beta_{i-1}}=\frac{(\sqrt{p_ip_{i+1}}+\rho_{i}\sqrt{(1-p_i)(1-p_{i+1})})(\sqrt{p_{i-1}p_{i}}+\rho_{i-1}\sqrt{(1-p_{i-1})(1-p_{i})})}{\sqrt{p_{i-1}p_{i+1}}}\leq 1,
\end{aligned}
\end{equation}
which is equivalent to inequality \eqref{alg_inequal}.

When $\rho_1=\rho_2=\dots=\rho_m=\rho$, denote $A$, $B$ and $C$ as above,
and we get
\begin{equation}
\begin{aligned}
\alpha_i\leq1 \iff
A\rho^2+B\rho-C\leq0
\Rightarrow
\rho\leq\frac{-B+\sqrt{B^2+4AC}}{2A}.
\end{aligned}
\end{equation}
Thus, given the marginal probabilities of a  variable as well as its neighbors', we can derive the constraint for the corresponding correlation in the Algorithm \ref{alg3}. The constraint is not relevant with other correlation coefficients. 

In the special case when $p_1=p_2=\dots=p_m=p$, the above constraint can be reduced to $\rho\leq\frac{\sqrt{p}}{1+\sqrt{p}}$, which is not related with the dimension $m$ and only related with the marginal probability $p$. Figure \ref{fig:mdep1} shows the relationship between the maximal allowed correlation $\rho_{max}$ and the marginal probability $p$. The maximal allowed correlation is monotonically increase with $p$ and converge to 0.5, which is consistent with the natural restrictions for positive definiteness discussed in Section 2.1.
\end{proof}

\begin{figure}[t!]
\centering
\includegraphics[width=0.8\textwidth]{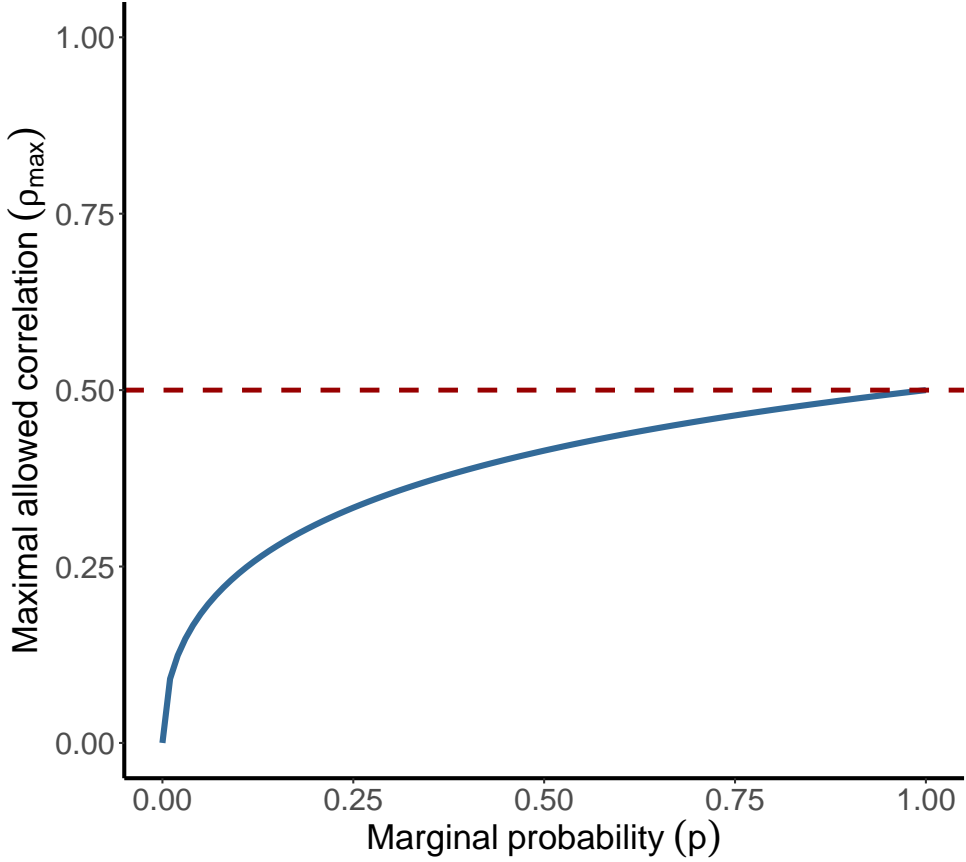}
\caption{Relationship between the maximal allowed correlation ($\rho_{max}$) and given marginal probability  ($p$) in Algorithm \ref{alg3}, when $p_1=p_2=\cdots=p_m=p$ and $\rho_1=\rho_2=\cdots=\rho_m=\rho$. The maximal correlation will monotonically increase as $p$ increases, and converge to 0.5.}
\label{fig:mdep1}
\end{figure}

%% Algorithm4
For Algorithm \ref{alg4}, we used the form $W_i=(1-U_i)Y_i+U_iY_{i-1}$ to generate variables with 1-dependent correlation structure, and then multiply $W_i$ with an $A_i$ to adjust the probability. Actually, the idea of multiplying an $A_i$ has been discussed in  Lunn and Davies' paper. As they discussed in the paper, this strategy can only be applied in $1$-dependent correlation structure. Nevertheless, we provide the details of this strategy for generating binary data with  $1$-dependent correlation structure in Algorithm \ref{alg4}. In the following, we provide the proof of related properties for Algorithm 4. 

\begin{algorithm}[!h]
\caption{Generate binary data under the  1-dependent correlation structure. (Method \uppercase\expandafter{\romannumeral2})} 
\label{alg4}
\hspace*{0.02in} {\bf Input:}
The expected values of the Bernoulli random variables $p_1,p_2,\cdots,p_m$; and the correlation coefficient vector $\bm{\rho}=(\rho_1,\rho_2,\cdots,\rho_{m-1})$ \\
\hspace*{0.02in} {\bf Output:}
The correlated binary variables $X_1,X_2,\cdots,X_m$
\begin{algorithmic}[1]
    \makeatletter
\setcounter{ALG@line}{-1}
\makeatother
\State Check whether the input satisfies the Prentice constraints.
\State $p_{max}=\max\{p_1,\cdots,p_m\}$
\State Denote $\alpha_i=\frac{p_i}{p_{max}}$, $ i=1,\cdots,m$
\State Denote $\rho_i'=\rho_i\frac{\sqrt{(1-p_i)(1-p_{i+1})}}{\sqrt{\alpha_i\alpha_{i+1}}(1-p_{max})}$, $i=1,\cdots,m-1$
\State $Y_1\sim Bern(p_{max})$
\State $W_1=Y_1$
\State $r_1=0$
\For{$i=2,\cdots,m$}
\State $r_i=\frac{\rho_{i-1}'}{1-r_{i-1}}$
\State \textbf{if} $r_i \textgreater 1$: Prompt the unavailability; \textbf{break}
\State Generate $U_i\sim Bern(r_i)$
\State Generate $Y_i\sim Bern(p_{max})$
\State $W_i=(1-U_i)Y_i+U_iY_{i-1}$
\EndFor
\For{$i=1,\cdots,m$}
\State Generate $A_i\sim Bern(\alpha_i)$
\State$X_i=A_iW_i$

\EndFor
\State\Return $X_1,X_2,\cdots,X_m$
\end{algorithmic}
\end{algorithm}

\begin{thm}
    If intermediate variables $\alpha_2$, $\cdots$, $\alpha_m$, $r_2$, $\cdots$, $r_m \in \left[0,1\right]$, Algorithm \ref{alg4} returns the binary data with given marginal expectation and   1-dependent correlation structure.
\end{thm}
\begin{proof}
From the definition, $EX_1=p_1$.
For $i$ in $\{2,\cdots,m\}$, we have
\begin{equation}
    EX_i=EA_iW_i=\alpha_i((1-r_i)p_{max}+r_ip_{max})
    =\frac{p_i}{p_{max}}p_{max}=p_i.
\end{equation}

Then for each $i \in \{1,2,\cdots,m-1\}$, we can obtain that
\begin{equation}
\begin{aligned}
    r_{i,i+1}=&\frac{cov(X_i,X_{i+1})}{\sqrt{Var(X_i)Var(X_{i+1})}}=
    \frac{\alpha_i \alpha_{i+1}cov((1-U_i)Y_i,U_{i+1}Y_{i})}{\sqrt{p_ip_{i+1}(1-p_i)(1-p_{i+1})}}\\
    =&\frac{\alpha_i \alpha_{i+1}(E((1-U_i)U_{i+1}Y_i^2)-E((1-U_i)Y_i)E(U_{i+1}Y_i))}{\sqrt{p_ip_{i+1}(1-p_i)(1-p_{i+1})}}\\
    =&\frac{\rho_i'(1-p_{max})\sqrt{p_ip_{i+1}}}{p_{max}\sqrt{(1-p_i)(1-p_{i+1})}}
    =\frac{\rho_i\sqrt{p_ip_{i+1}}}{p_{max}\sqrt{\alpha_i\alpha_{i+1}}}
    =\frac{\rho_i\sqrt{p_ip_{i+1}}}{p_{max}\sqrt{\frac{p_ip_{i+1}}{p_{max}^2}}}=\rho_i.
\end{aligned}
\end{equation}

Meanwhile, it is obvious that $r_{i,j}=0$, $|i-j|\geq 2$. We obtain the correctness of the algorithm.
\end{proof}

The data generated by Algorithm \ref{alg4} need the intermediate parameters $\gamma_i$ to locate in $[0,1]$. In the following theorem, we show the restriction for $\pmb{\rho}$ such that the intermediate parameter requirement is satisfied.

\begin{thm}
\label{thm2.8}
Only binary data with non-negative  1-dependent correlation structure satisfying the following inequalities can be generated from Algorithm \ref{alg4}:
\begin{equation}
    0\leq\rho_i \leq \sqrt{\frac{p_ip_{i+1}}{(1-p_i)(1-p_{i+1})}}\frac{1-p_{max}}{p_{max}}(1-r_{i}),
\end{equation}
where $r_i=\frac{\rho_{i-1}}{1-r_{i-1}}\sqrt{\frac{(1-p_{i-1})(1-p_i)}{p_{i-1}p_i}}\frac{p_{max}}{1-p_{max}}$
is a function of previous correlation $\rho_{i-1}$. In the special case when $p_1=p_2=\dots=p_m=p$ and $\rho_1=\rho_2=\dots=\rho_m=\rho\geq0$, the applicable condition is $\frac{1}{\sqrt{1-4\rho}}\left[(1+\sqrt{1-4\rho})^{m+1}-(1-\sqrt{1-4\rho})^{m+1}\right]\geq 0$.

\end{thm}
\begin{proof}
In Algorithm \ref{alg4}, $A_i$ and $Y_i$ can always be generated since their corresponding probabilities $\alpha_i$ and $p_{max}$ are in the range of $[0,1]$. Hence, we only need to consider the generation of $U_i$.

Based on the definition of $U_i$'s marginal probability $r_i$, it is obvious that $r_i\geq 0$ if $r_{i-1} \leq 1$ holds. Therefore, the only requirement of Algorithm \ref{alg4} is $r_i\leq 1$ ($i=1,\dots,m$). Thus, we have the applicable condition:
\begin{equation}
    r_i= \frac{\rho_{i-1}}{1-r_{i-1}}\sqrt{\frac{(1-p_{i-1})(1-p_i)}{p_{i-1}p_i}}\frac{p_{max}}{1-p_{max}} \leq1 \iff \rho_i \leq \sqrt{\frac{p_ip_{i+1}}{(1-p_i)(1-p_{i+1})}}\frac{1-p_{max}}{p_{max}}(1-r_{i}).
\end{equation}
In practice, we can obtain the condition for each $\rho_i$ by iteration.  

In the special case when $p_1=\cdots=p_m=p$, $\rho_1=\cdots=\rho_m=\rho$,
the requirement becomes the following series of inequalities:
\begin{equation}\label{rho4}
\begin{aligned}
&r_1=0\leq1\\
&r_2=\frac{\rho}{1-r_1}=\rho\leq1\\
&r_3=\frac{\rho}{1-r_2}=\frac{\rho}{1-\rho}\leq1
\Rightarrow \rho\leq\frac{1}{2}\\
&r_4=\frac{\rho}{1-r_3}=\frac{\rho}{1-\frac{\rho}{1-\rho}}=\frac{\rho(1-\rho)}{1-2\rho}\leq1
\Rightarrow \rho\leq \frac{3-\sqrt{5}}{2}\\
&\cdots
\end{aligned}
\end{equation}
The general term formula of $r_i$ is
$r_i=2\rho\left[\frac{(1+\sqrt{1-4\rho})^{i-1}-(1-\sqrt{1-4\rho})^{i-1}}{(1+\sqrt{1-4\rho})^i-(1-\sqrt{1-4\rho})^i}\right]$.
The series of inequalities can be reduced to 
\begin{equation}
    \frac{1}{\sqrt{1-4\rho}}\left[(1+\sqrt{1-4\rho})^{i+1}-(1-\sqrt{1-4\rho})^{i+1}\right]\geq 0,                                      \end{equation}
where $i=1,\dots, m$.

From the inequalities above  we see that as $i$ increases, the restriction of $\rho$ becomes more and more strict. Hence, we only need to satisfy the last   inequality in Algorithm \ref{alg4}, i.e.,
\begin{equation}
    \frac{1}{\sqrt{1-4\rho}}\left[(1+\sqrt{1-4\rho})^{m+1}-(1-\sqrt{1-4\rho})^{m+1}\right]\geq 0.
\end{equation}
\end{proof}

Compared to the applicable condition of Algorithm \ref{alg3}, the restriction of Algorithm  \ref{alg4} is only related with the dimension $m$ and not  the marginal probability $p$. 
Figure \ref{fig:mdep0} describes the relationship between the maximal allowed correlation $\rho_{max}$ and the dimension $m$. The $\rho_{max}$  monotonically decreases as $m$ increases, and converges to $0.25$. Therefore, if $\rho\leq 0.25$, this algorithm is suitable to generate binary data with an arbitrary dimension  $m$.

\begin{figure}[h]
\centering
\includegraphics[width=0.8\textwidth]{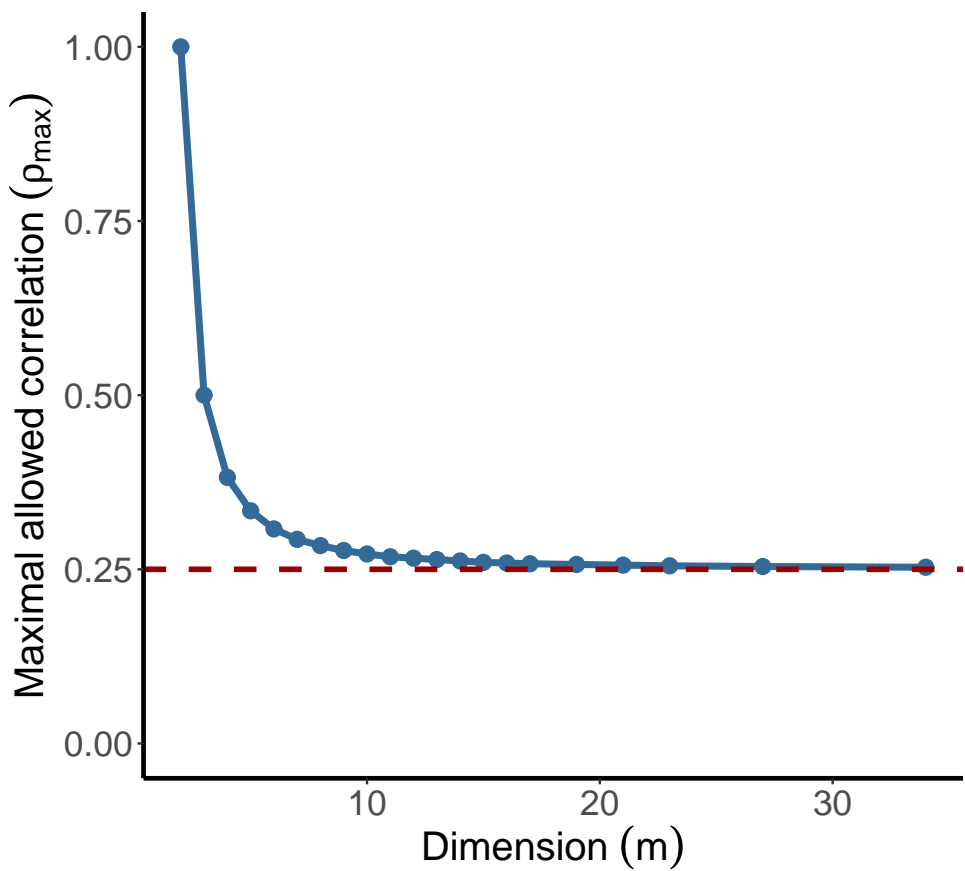}
\caption{Relationship between the maximal allowed correlation $\rho_{max}$ and the dimension $m$ in Algorithm \ref{alg4}, when $p_1=p_2=\dots=p_m=p$ and $\rho_1=\rho_2=\dots=\rho_m=\rho$. As $m$ increases, the maximum of $\rho$ will decrease, and gradually convergent to $0.25$.}
\label{fig:mdep0}
\end{figure}

To conclude, both algorithms have their own specialities and applicable conditions for generating binary data with  $1$-dependent correlation structure. In Algorithm \ref{alg3}, the limitation for each entry of correlation vector is only related with nearby marginal probabilities. In Algorithm \ref{alg4}, the limitation is also related with previous entries of correlation vector. It is difficult to distinguish which algorithm has more general applicable conditions. We give a detailed analysis when $p_1=p_2=\dots=p_m=p$ and $\rho_1=\rho_2=\dots=\rho_m=\rho$.
In this situation, Algorithm \ref{alg3} has no limitation on the dimension, meaning that the dimension can be arbitrarily large with feasible marginal probabilities. Thus, this algorithm is perfect for the situation when  high dimension is required. On the other hand, Algorithm \ref{alg4} is more flexible when the dimension is not too high, since it is not restricted by the marginal probabilities.

In order to incorporate  the two methods, we provide a function  {cBern1dep} in our R package  {CorBin}, which can automatically choose the suitable algorithm based on the given $\bm{p}$ and $\bm{\rho}$. In our function, we will first derive $r_i$ ($i=1,2,\cdots,m$) in Algorithm \ref{alg4}. 
If all the $r_i$ lie in the interval $\left[0,1\right]$, we will use Algorithm \ref{alg4} to generate the binary data. If not, the function will automatically call function  {rhoMax1dep} to calculate the largest  $\bm{\rho}$ allowed in Algorithm \ref{alg3}. If the given $\bm{\rho}$ lies in the interval, the binary data will be generated using Algorithm \ref{alg3}.

\subsection{$K$-dependent correlation structure and general correlation matrices}
\label{general}
In Section \ref{statm} we provide two algorithms to generate binary data with   1-dependent correlation structure.
Here we discuss the generation of the binary data with   $K$-dependent ($K> 1$) correlation structure by extending Algorithm \ref{alg3}. Specifically, if we set $K=m-1$, we can obtain binary data with  the general non-negative correlation matrices. Based on the intuition of Algorithm \ref{alg3}, we provide the details of the binary data generation algorithm under the  $K$-dependent correlation structure (and also a general correlation matrix) in Algorithm \ref{alg5}. 

We first denote $\bm{Y}=(Y_1,\cdots,Y_m)$ as a $K \times m$ matrix:
\begin{equation}
\bm{Y}={
\left( \begin{array}{cccc}
Y_{11} & Y_{12} & \cdots&Y_{1m}\\
Y_{21} & Y_{22}  &\cdots&Y_{2m}\\
\vdots&\vdots&\ddots&\vdots\\
Y_{K1}&Y_{K2}&\cdots&Y_{Km}
\end{array} 
\right )}.
\end{equation}
In the algorithm, we use $\bm{\rho_i}$ to denote the $m-i$ elements on the $i$-th diagonal of the correlation matrix, i.e.,
\begin{equation}
    \begin{aligned}
        \bm{\rho_i}=(\rho_{i1},\rho_{i2},\cdots,\rho_{i(m-i)})=(r_{1(i+1)},r_{2(i+2)},\cdots,r_{(m-i)m}).
    \end{aligned}
\end{equation}

\begin{algorithm}[!h]
\caption{Generate binary data under  $K$-dependent correlation structure  (and general correlation matrix if $K=m-1$).} 
\label{alg5}
\hspace*{0.02in} {\bf Input:}
The expected values of the Bernoulli random variables $p_1,p_2,\cdots,p_m$; and the correlation coefficient vector $\bm{\rho_1}, \bm{\rho_2},\cdots,\bm{\rho_K}$ \\
\hspace*{0.02in} {\bf Output:}
The correlated binary variables $X_1,X_2,\cdots,X_m$
\begin{algorithmic}[1]
    \makeatletter
\setcounter{ALG@line}{-1}
\makeatother
\State Check whether the input satisfies the Prentice constraints.
\State Let $p_{m+1}=p_{m+2}=\cdots=p_{m+K}=p_m$
\For{$i=1,2,\cdots,K$}
\State Let $\rho_{i,(m-(i-1))}=\rho_{i,(m-(i-1)+1)}=\cdots=\rho_{i,m}=0$
\For{$j=1,2,\cdots,m$}
\State 
$\beta_{ij}=\frac{p_jp_{i+j}}{p_jp_{i+j}+\rho_{ij}\sqrt{p_jp_{i+j}(1-p_j)(1-p_{i+j})}}$
\State
    Generate $Y_{ij}\sim Bern(\beta_{ij})$
\EndFor
\EndFor
\State 
$\alpha_1=\frac{p_1}{\prod_{l=1}^{K}\beta_{li}}$
\State \textbf{if} $\alpha_1 \textgreater 1$: Prompt the unavailability.
\State Generate $U_1\sim Bern(\alpha_1)$
\State $X_1=U_1\prod_{l=1}^{K}Y_{li}$
\For{$i=2,\cdots,m$}
\State Denote $K_i'=\min \{i-1,K\} $
\State
$\alpha_i=\frac{p_i}{\prod_{l=1}^{K}\beta_{li}\prod_{l=1}^{K_i'}\beta_{l(i-l)}}$
\State \textbf{if} $\alpha_i \textgreater 1$: Prompt the unavailability; \textbf{break}
%=\beta_{1i}\beta_{1(i-1)}\beta_{2i}\beta_{2(i-2)}\cdots\beta_{(i-1)i}\beta_{(i-1)1}}$
\State
Generate $U_i\sim Bern(\alpha_i)$
\State
$X_i=U_i\prod_{l=1}^{K}Y_{li}\prod_{l=1}^{K_i'}Y_{l(i-l)}$
\EndFor
\State\Return $X_1,X_2,\cdots,X_m$
\end{algorithmic}
\end{algorithm}

In the following, we show the justification of Algorithm \ref{alg5} in Theorem \ref{thm:kdep}.

\begin{thm}\label{thm:kdep} 
    If intermediate variables $\alpha_1$, $\cdots$, $\alpha_m$, $\beta_{11}$, $\cdots$, $\beta_{Km} \in \left[0,1\right]$, Algorithm \ref{alg5} returns the corresponding binary data  with given  marginal probabilities and  $K$-dependent correlation structure.
\end{thm} 

\begin{proof}
    For $i=1$, we have
    \begin{equation}
    EX_1=EU_1\prod_{l=1}^{K}Y_{li}=\alpha_1\prod_{l=1}^{K}\beta_{li}=p_1.
    \end{equation}
For each $i \in \{2,\cdots,m\}$, 
\begin{equation}
    EX_i=EU_i\prod_{l=1}^{K}Y_{li}\prod_{l=1}^{K_i'}Y_{l(i-l)}=\alpha_i\prod_{l=1}^{K}\beta_{li}\prod_{l=1}^{K_i'}\beta_{l(i-l)}=\frac{p_i}{\prod_{l=1}^{K}\beta_{li}\prod_{l=1}^{K_i'}\beta_{l(i-l)}}\prod_{l=1}^{K}\beta_{li}\prod_{l=1}^{K_i'}\beta_{l(i-l)}=p_i.
\end{equation}
From the generation process, it is obvious that $r_{ij}=0$ when $ \left|j-i\right|>K$.
Considering $i=1$, $j \in \{1, 2, \cdots, m-1\}$, 
\begin{equation}
    \begin{aligned}
        r_{1,1+j}=&\frac{cov(X_1,X_{1+j})}{\sqrt{Var(X_1)Var(X_{1+j})}}=\frac{cov(U_1\prod_{l=1}^{K_1}Y_{l1}, U_{1+j}\prod_{l=1}^{K}Y_{l(1+j)}\prod_{l=1}^{K_{1+j}'}Y_{l(1+j-l)}}{\sqrt{p_1p_{1+j}(1-p_1)(1-p_{1+j})}}\\
        =&\frac{E(U_1U_{1+j}\prod_{l=1}^{K}Y_{l1}\prod_{l=1}^{K}Y_{l(1+j)}\prod_{l=1}^{K_{1+j}'}Y_{l(1+j-l)})}{\sqrt{p_1p_{1+j}(1-p_1)(1-p_{1+j})}}\\
        &-\frac{E(U_1\prod_{l=1}^{K}Y_{l1})E(U_{1+j}\prod_{l=1}^{K}Y_{l(1+j)}\prod_{l=1}^{K_{1+j}'}Y_{l(1+j-l)})}{\sqrt{p_1p_{1+j}(1-p_1)(1-p_{1+j})}}\\
       % =&\frac{\alpha_{i}\alpha_{i+1}\beta_{i-1}\beta_i\beta_{i+1}(1-\beta_i)}{\sqrt{p_ip_{i+1}(1-p_i)(1-p_{i+1})}}
        =&\frac{p_{1}p_{1+j}(1-\beta_{j1})}{\beta_{j1}\sqrt{p_1p_{1+j}(1-p_1)(1-p_{1+j})}}=\rho_{1j}.
    \end{aligned}
    \end{equation}
For $i \in \{2, \cdots,  m-1\}$, $j \in \{1, 2, \cdots, m-i\}$, we have
\begin{equation}
\begin{aligned}
    r_{i,i+j}=&\frac{cov(X_i,X_{i+j})}{\sqrt{Var(X_i)Var(X_{i+j})}}=\frac{cov(U_i\prod_{l=1}^{K}Y_{li}\prod_{l=1}^{K_i'}Y_{l(i-l)},
    U_{i+j}\prod_{l=1}^{K}Y_{l(i+j)}\prod_{l=1}^{K_{i+j}'}Y_{l(i+j-l)})}{\sqrt{p_ip_{i+j}(1-p_i)(1-p_{i+j})}}\\
    =&\frac{E(U_iU_{i+j}\prod_{l=1}^{K}Y_{li}\prod_{l=1}^{K_i'}Y_{l(i-l)}\prod_{l=1}^{K}Y_{l(i+j)}\prod_{l=1}^{K_{i+j}'}Y_{l(i+j-l)})}{\sqrt{p_ip_{i+j}(1-p_i)(1-p_{i+j})}}\\
    &-\frac{E(U_i\prod_{l=1}^{K}Y_{li}\prod_{l=1}^{K_i'}Y_{l(i-l)})E(U_{i+j}\prod_{l=1}^{K}Y_{l(i+j)}\prod_{l=1}^{K_{i+j}'}Y_{l(i+j-l)})}{\sqrt{p_ip_{i+j}(1-p_i)(1-p_{i+j})}}\\
   % =&\frac{\alpha_{i}\alpha_{i+1}\beta_{i-1}\beta_i\beta_{i+1}(1-\beta_i)}{\sqrt{p_ip_{i+1}(1-p_i)(1-p_{i+1})}}
    =&\frac{p_{i}p_{i+j}(1-\beta_{ji})}{\beta_{ji}\sqrt{p_ip_{i+j}(1-p_i)(1-p_{i+j})}}=\rho_{ij}.
\end{aligned}
\end{equation}
Here we have finished the proof.
\end{proof}

Due to the increased model complexity, it is difficult to derive the applicable condition of Algorithm \ref{alg5} theoretically. However, given marginal probabilities and a general correlation matrix, we can still check whether the binary data can be generated using the algorithm by examining whether all intermediate parameters $\alpha_1,\cdots,\alpha_m$, $\beta_{11},\cdots,\beta_{Km}$ lie in the range of $[0,1]$.

\section{Performance} \label{3}
We implemented and integrated the above mentioned algorithms in an R package {CorBin}. In this section, we mainly demonstrate the effectiveness and computational  efficiency of our package. 
If a data set is generated from the desired distribution, the sample mean should  converge to the specified marginal probabilities when sample size increases. Meanwhile, the sample correlation matrix should also  converge to the specified correlation matrix. Here, we demonstrate the effectiveness of our package by checking the consistency of sample mean and correlation matrix from the  generated data. After that, we  demonstrate the computational efficiency of our package by calculating the time needed for generating large-scale high-dimensional datasets. We further compare computational time with two commonly used binary data generation packages: {bindata} (Leisch et al. 1998) and {MultiOrd} (Demirtas 2006).

\subsection{Effectiveness}
In order to check the consistency of sample mean and correlation matrix, we generate datasets with different sample sizes in which the dimension $m$ is fixed to 100. For sample mean, we use the $l_2$ norm of the difference between the sample mean and the specified marginal probabilities as the error. For correlation matrix, we calculate the Frobenius norm of the residual matrix between sample and desired correlation matrix as the error.  
We randomly sampled the marginal expectations from a uniform distribution $U(0.5,0.8)$. The upper bound of correlation coefficients based on the Prentice constraints is $\sqrt{\frac{0.5\times0.2}{0.5\times0.8}}=0.5$. Then we generated a correlation coefficient from a uniform distribution $U\left(0,0.5\right)$ for exchangeable and AR(1) correlation structures. As the constraints for 1-dependent correlation structures are more stringent, we simulated the correlation coefficient from $U\left(0,0.2\right)$. Although Algorithm 5 can be applied to general cases, randomly generating the correlation coefficients cannot always satisfy the natural restrictions. Thus, without loss of generality, we fixed the structures to AR(1), and the settings were the same with simulations of Algorithm 2. We ran 10 times of simulations and calculated the average of errors for each distribution and verify the effectiveness of the algorithms. Figure \ref{fig:precision} shows that under the four correlation structures we have considered  and a specified general correlation matrix, both errors gradually approached to $0$ as sample size increased, indicating that the sample mean and correlation matrix of the generated data converged to the true settings we specified. These results demonstrate the effectiveness of our methods.

\begin{figure}[htpb]
\centering
\includegraphics[width=\linewidth]{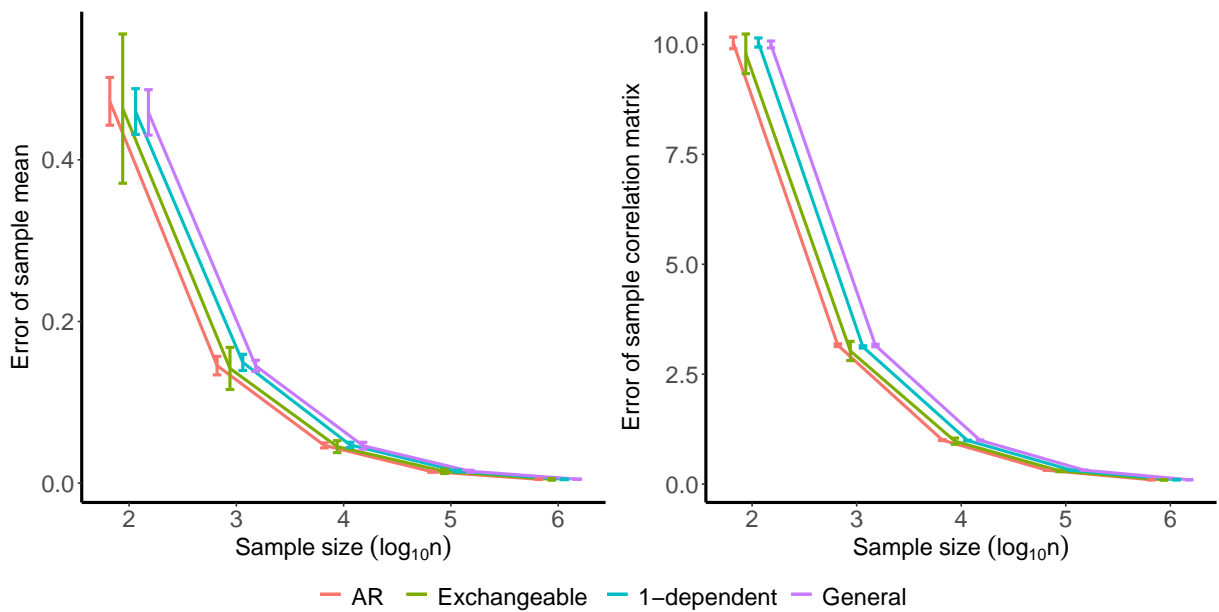}
\caption{As sample size increases, the sample mean and correlation matrix converge to the specified marginal probabilities and correlation matrix.}
\label{fig:precision}
\end{figure}

In addition, we provide five simple examples to illustrate the constructions of the algorithms and the choices for  the parameters for better demonstration. The specified marginal probabilities and the specified correlation for each algorithms are summarized in Table \ref{tab:ex}. The details of the data generation process for each example are attached in Supplementary Material (Example1-5.csv). Besides, we also provide their reproducing code in Supplementary Material (Example-code.R).

\begin{table}[!htbp]
    \centering
    \bcaption{\bf{Examples for illustration of the algorithms}}{ The table summarized specified marginal probabilities and correlation matrices  for each example, The details of data generation process for each example were presented in the corresponding file. }\label{tab:ex}
    \begin{tabular}{cccccc}
    \toprule
    Algorithm& Structure&$m$&$\bm{p}$&$\bm{R}$&Supplementary File\\
    \hline
    \ref{alg1}&Exchangeable& 3&$(0.1,0.2,0.3)$&
    ${
    \left( \begin{array}{ccc}
    1 & 0.3 & 0.3\\
    0.3  & 1 & 0.3  \\
    0.3  & 0.3  & 1
    \end{array} 
    \right )}$&Example1.csv\\ \hline
    \ref{alg2} & AR(1)&  3&$(0.1,0.2,0.3)$&
    ${
    \left( \begin{array}{ccc}
    1 & 0.2 & 0.1\\
    0.2  & 1 & 0.5  \\
    0.1  & 0.5  & 1
    \end{array} 
    \right )}$&Example2.csv\\ \hline
    \ref{alg3} &$1$-dependent&  3&$(0.80,0.82,0.83)$&
    ${
    \left( \begin{array}{ccc}
    1 & 0.3 & 0\\
    0.3  & 1 & 0.5  \\
    0  & 0.5  & 1
    \end{array} 
    \right )}$&Example3.csv\\ \hline
    \ref{alg4} & $1$-dependent&  3&$(0.80,0.82,0.83)$&
    ${
    \left( \begin{array}{ccc}
    1 & 0.3 & 0\\
    0.3  & 1 & 0.5  \\
    0  & 0.5  & 1
    \end{array} 
    \right )}$&Example4.csv\\ \hline
    \ref{alg5} & Generalized&  3&$(0.6,0.7,0.8)$&
    ${
    \left( \begin{array}{ccc}
    1 & 0.3 & 0.1\\
    0.3  & 1 & 0.2  \\
    0.1  & 0.2  & 1
    \end{array} 
    \right )}$&Example5.csv\\ \hline
    \bottomrule
    \end{tabular}
    \end{table}

\subsection{Computational efficiency}

In this section, we demonstrate the superiority of our package in computational efficiency. All experiments performed here were based on a single processor of an  Intel(R) Core(TM) 2.20GHz PC. For comparison, we also considered two commonly used packages {bindata} and {MultiOrd} to generate high dimensional binary data in the experiments. 

It is easy to find that all algorithms presented in Section 2.2, 2.3 and 2.4 involve only one layer of iterative process. Hence, the time complexity of our algorithms generating binary data with exchangeable, decaying-product and  1-dependent correlation structures is linear with respect to dimension $m$ theoretically.
In Section 2.5, although there are two iteration layers in Algorithm \ref{alg5}, the time complexity is still linear with respect to $m$ if $K$ is irrelevant with $m$, which is normal in ordinary $K$-dependent correlation structure. However, when we want to generate binary data with a general correlation matrix, $K$ will be specified to $m-1$ and the time complexity will become quadratic with respect to $m$. Figure \ref{fig:cbtime} presents the average time for generating binary data with different dimensions using {CorBin}, which further validates the linear time complexities of Algorithm 1-4 and quadratic time complexity of Algorithm 5. Here we use Algorithm 5 to generate binary data with autoregressive structure in simulation experiments. Please refer to Supplementary Table S1 for the numeric details of average time with a more general range of $m$ ($m=10^2\sim 10^6$).

\begin{figure}[t!]
\centering
\includegraphics[width=\textwidth]{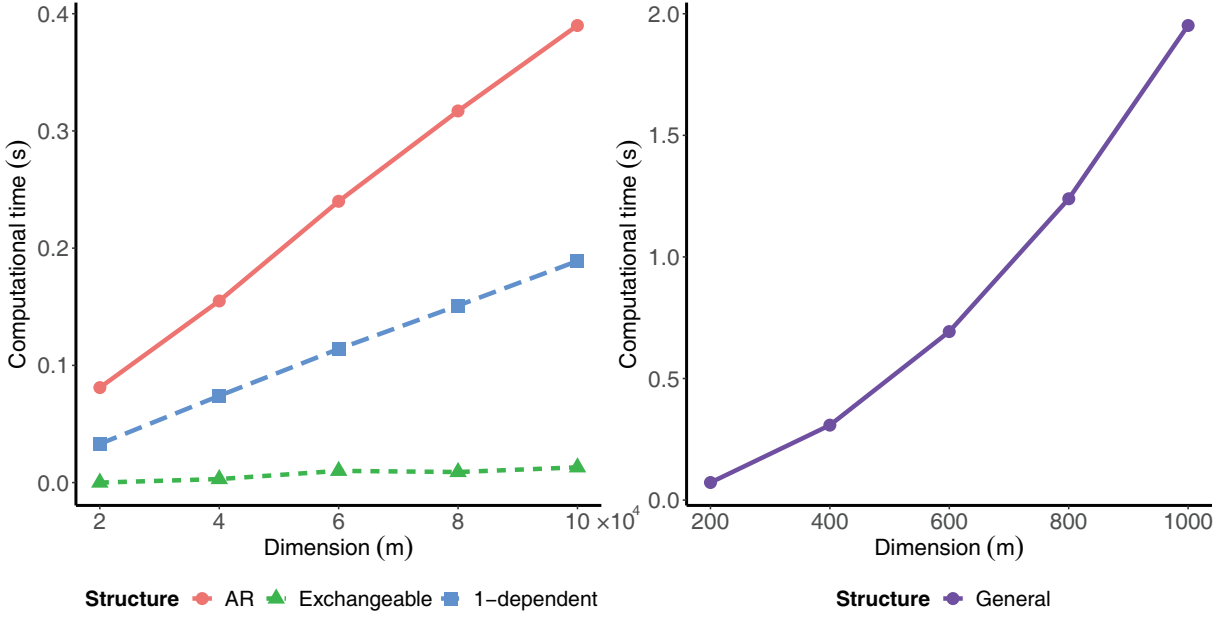}
\caption{The  relationship between the computational time  (in seconds) and the dimension when simulating binary data using our package. The time is the average of 10 runs.}
\label{fig:cbtime}
\end{figure}

The calculation efficiency is  impressive when generating the high-dimensional binary data. It takes only 0.2, 4.0 and 2.0 seconds to generate a $10^6$-dimensional binary data with exchangeable correlation structure, decaying-product and 1-dependent correlation structures, respectively. This dimension scale is  too high for other data generation packages, such as {bindata} and {MultiOrd}. Table \ref{tab:time} shows the  time of different packages for generating binary data in a relative small scale ($m=100\sim 500$). The time recorded is based on experiments of 10 runs. For data generation with general correlation matrix, our algorithm is still very efficient compared to other packages.

\begin{table}[!htbp]
\centering
\bcaption{\bf{Computational times (in seconds) needed for three algorithms under different correlation structures. }}{The least time under each condition is highlighted in boldface. Standard deviations are in the bracket.}\label{tab:time}
\begin{tabular}{ccccc}
\toprule
Structure& $m$&CorBin&bindata&MultiOrd\\
\hline
\multirow{3}{*}{Exchangeable ($\rho=0.5$)}& 100&\bm{$2e-5$} $(6e-6)$ &14.19 (0.07) &6.517 (0.06)\\
 &200&\bm{$4e-5$} $(5e-6)$ &57.13 (0.38)&26.24 (0.63)\\
 &500&\bm{$8e-5$} $(7e-6)$ &360.3 (2.46)&164.9 (0.47)\\% this line need change
\hline
 \multirow{3}{*}{AR(1) ($\rho=0.4$)}&100&\bm{$5e-4$} $(4e-5)$ &14.36 (0.06)&5.783 (0.05)\\
 &200&\bm{$1e-3$} $(7e-5)$ &56.27 (0.93)&18.18 (0.21)\\% this line need change 
 &500&\bm{$3e-3$} $(3e-4)$ &358.3 (2.24)&123.5 (0.31)\\
 \hline
 \multirow{3}{*}{1-dependent ($\rho=0.2$)}&100&\bm{$3e-4$} $(2e-5)$ &14.30 (0.05)&5.977 (0.06)\\
 &200& \bm{$6e-4$} $(1e-4)$ & 56.83 (0.28)& 18.81 (0.33)\\ 
  &500& \bm{$1e-3$} $(9e-5)$ & 366.2 (3.27)& 112.3 (0.13)\\ 
\bottomrule
\end{tabular}
\end{table}

It can be seen from Table \ref{tab:time} and Figure 4, our method {CorBin} scales linearly  with the dimension  under three correlation structures and scales quadratically with general correlation matrix. As a comparison, the computational time increases rapidly with the growth of dimensions for the other two packages. Taking the exchangeable correlation structure as an example, it takes {bindata} around 14.2 seconds to generate a 100-dimensional data  and 360.3 seconds to generate a 500-dimensional data, which is around 25 times of the former. The results are similar  for {MultiOrd}. Moreover, regardless of the increasing rate with the dimension, our method has significant superiority over the other methods. When the dimension is 100, the time {CorBin} used is around $\frac{1}{700,000}$ of {bindata} and $\frac{1}{320,000}$ of {MultiOrd} in exchangeable correlation structure, and also far less than the other two in AR(1) and 1-dependent structure. When the dimension grows to 500, the advantage is even more obvious, with around $\frac{1}{4,500,000}$ of {bindata} and $\frac{1}{2,060,000}$ of {MultiOrd} in exchangeable correlation structure. For generating data with general correlation matrices, the ratios become around $\frac{1}{680}$ of {bindata} and $\frac{1}{320}$ of {MultiOrd}, respectively.  Figure 5 shows the comparison among {CorBin}, {bindata} and  {MultiOrd} in terms of computational time for general cases.

\begin{figure}[t!]
\centering
\includegraphics[width=0.8\textwidth]{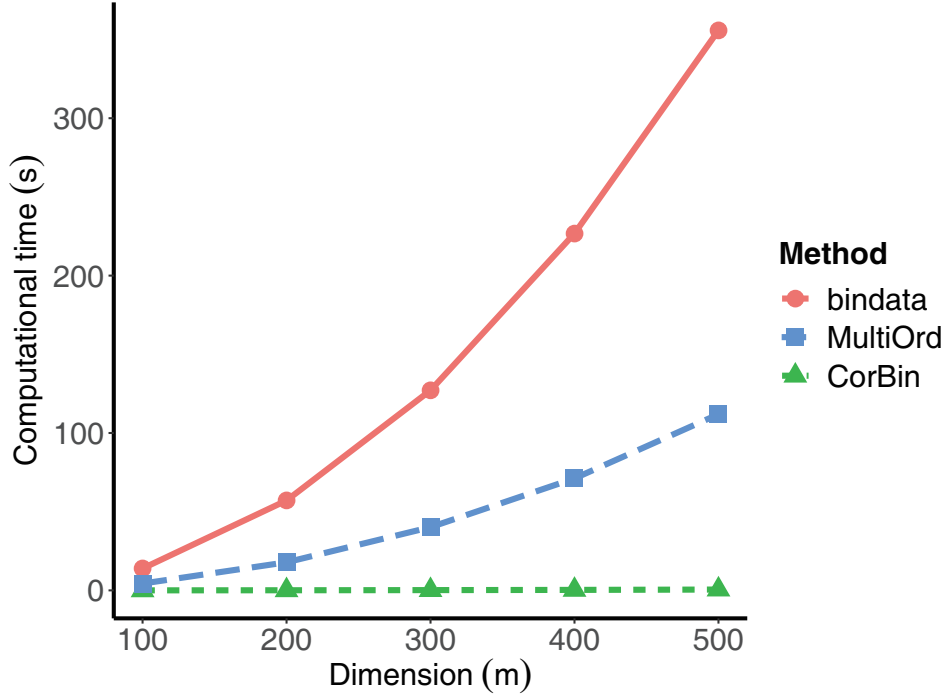}
\caption{The  relationship between the computational time (in seconds) and the dimension when simulating binary data with general correlation structure using {CorBin}, {bindata} and {MultiOrd}. The time is the average of 10 runs.}
\label{fig:method_compare}
\end{figure}

\section{Discussion and conclusions} \label{6}
In this article, we have proposed several efficient algorithms to generate high-dimensional correlated binary data with varied marginal expectations  and  correlation structures. We first focus on three common correlation structures including exchangeable, decaying-product as well as $K$-dependent
correlation structures, and then generalize the method on $K$-dependent structure and extend the applicability to any non-negative correlation matrices.
An R package {CorBin} is also built based on these algorithms and uploaded on CRAN for readers to use. 
Compared with two state-of-the-art binary data generation  packages {bindata} and {MultiOrd} \citep{leisch1998generation,demirtas2006method}, our algorithms require no complicated numerical procedures such as equation-solving or numerical integration and have linear time complexity with respect to the dimension when generating binary data with common correlation structures, leading to significant improvement in computational efficiency. In our simulations, {CorBin} needs less than 0.002 seconds to generate a $10\times1000$-dimensional binary data with exchangeable correlation structure, while generating such data takes more than 14,000 seconds and 6,400 seconds for {bindata} and {MultiOrd}, respectively.

Compared with Lunn and Davies' method, we generalize  the algorithms so that the unequal probability settings can be satisfied. Concretely speaking, Lunn and Davies actually generated clusters of binary variables, and specified fixed the marginal probability and correlation coefficient in each independent cluster. Thus, it is not feasible to specify unequal probabilities in their method. Specifically, for exchangeable correlation structures, we generated each variable by $X_i=\left(1-U_i\right)Y_i+U_iZ$. Lunn and Davies first set a fixed probability $p$ and generate independent $Z\sim Bern\left(p\right)$, $Y_i\sim Bern\left(p\right)$ and $U_i\sim Bern\left(\sqrt\rho\right)$. In order to generate variables with unequal probabilities, an intuitive way is to simply fix a $\gamma$ and generate $Z\sim Bern\left(\gamma\right)$, and adjust the expectation of $Y_i$ so that the probability of $X_i$ is $p_i$. It is infeasible because there is no way to guarantee the expectation of $Y_i$ is exactly lies in $\left[0,1\right]$. A subtle construction of $\gamma$, $\alpha_i$ and $\beta_i$ is important in our algorithm, which can obtain proper $\alpha_i$s $\in\left[0,1\right]$, and derive the desired result (as proved in Theorem 2.1 and 2.2).
For AR(1) correlation structures, we generated each variable by $X_i=\left(1-U_i\right)Y_i+U_iX_{i-1}$, while  Lunn and Davies generated $Y_i\sim Bern\left(p\right)$ and $U_i\sim Bern\left(r_i\right)$. Thus, the expectation of $X_i$ is dependent on the expectation of $X_{i-1}$, $U_i$, and $Y_i$, and making the construction of parameters untrivial. We provide a recursive method to generate the probability of those mediating variables, guaranteeing the feasible of the algorithm (Theorem 2.3 and 2.4).
For $1$-dependent correlation structures, we provide two algorithms. Algorithm 4 generalized  Lunn and Davies' method and Algorithm 3 was unrelated with  Lunn and Davies' method. We thoroughly studied two algorithms and derived their usage scopes in the manuscript. As discussed in  Lunn and Davies' paper, Algorithm 4 cannot be applied to $K>1$ situation. Our Algorithm 3 made up for this drawback. Notably, Algorithm 3 can be generalized to general cases with unequal probabilities and unequal correlation coefficients (Section 2.5, Algorithm 5).

There are still some limitations with our methods. First, our package is applicable only when the correlations are non-negative, because  in our algorithms we need to generate some variables following Bernoulli distribution with marginal probabilities related to the correlations. Negative correlations will lead to negative marginal probabilities, making it infeasible to generate corresponding binary vectors.
Although  in most situations, the capacity to generate binary data with positive correlations will suffice \citep{preisser2014comparison}, 
negative correlations may 
still arise in some special situations. For these situations, \cite{guerra2014note} proposed an alternative way to simulate discrete random vectors with decaying product structure, in which negative correlations are allowed.

We further derived applicable conditions for Algorithm 1-5, and surprisingly found that if the Prentice constraints are satisfied, our algorithms will be able to generate any specified binary data with non-negative exchangeable and decaying-product structures. But this is not the case for the $K$-dependent stationary structure. Therefore, we proposed two algorithms with different applicable conditions to generate binary data for $K=1$ and an algorithm with $K>1$. Our package will  automatically select the suitable algorithm according to the input parameters, but an algorithm with more general applicable conditions is still needed for  $K$-dependent structure.

%\section{Conclusion}
%\label{sec:conc}

\beginsupplement
\section*{Appendix A}
\begin{comment}
\subsection*{Supplementary Tables}\label{7}
\begin{figure}[H]
\centering
\includegraphics[width=\linewidth]{Fig_S1.eps}
\caption{Effectiveness of Algorithm 5. As sample size increases, the sample mean and correlation matrix converge to the specified marginal probabilities and correlation matrix.}
\label{fig:precision0}
\end{figure}
\end{comment}
\begin{table}[H]
\centering
\bcaption{\bf{The computational time (in seconds) consumed in CorBin for generating binary data with different correlation structures.}}{ The time is the average of 10 runs. Standard deviations are in the brackets.}\label{tab:addm}
\begin{tabular}{cccccc}
\toprule
$m$& $1e+2$&$1e+3$&$1e+4$&$1e+5$&$1e+6$\\
\midrule
 \multirow{2}{*}{Exchangeable}  &   $2e-5$&   $2e-4$   &   $2e-3$& $2e-2$ &$2e-1$\\
 &($5e-6$)  & ($5e-5$)&($3e-4$)& ($5e-3$)&($1e-2$)\\
 \hline
\multirow{2}{*}{AR(1) }& $5e-4$ &$3e-3$& $4e-2$& $4e-1$& $4e+0$\\
 &($4e-5$) &($7e-4$) &($8e-3$) &($2e-2$)
&($7e-2$)\\
\hline
\multirow{2}{*}{1-dependent}& $2e-4$& $2e-3$ &$2e-2$& $2e-1$ &$2e+0$\\
 & ($2e-5$) & ($1e-4$) & ($3e-3$) & ($2e-2$)
&($4e-2$)\\
\bottomrule
\end{tabular}
\end{table}

\bigskip
\begin{center}
{\large\bf SUPPLEMENTARY MATERIAL}
\end{center}

\begin{description}

%\item[Title:] Brief description. (file type)

\item[CorBin:] R-package CorBin containing code to implement the algorithms described in the article. (GNU zipped tar file)
\item[CorBin-manual:] User manual for R package CorBin. (.pdf file).
\item[Examples:] The demonstration data contain five CSV files (Example1-5.csv), corresponding to five examples in described in Section 3.1 (Table 1), which illustrate the constructions of the algorithms and the choices for the parameters. (.rar file)
\item[Example-code:] The reproducing code for demonstration data. (.R file).

\end{description}

\section*{Acknowledgements}
 We thank  the anonymous reviewer and the editor for their highly constructive and detailed feedback that helped us improve our manuscript substantially.

\section*{Funding}
This research was supported in part by the NSF grant DMS 1713120.

\bibliographystyle{agsm}
\bibliography{refs}

@article{pritchard2001linkage,
  title={Linkage disequilibrium in humans: models and data},
  author={Pritchard, Jonathan K and Przeworski, Molly},
  journal={The American Journal of Human Genetics},
  volume={69},
  number={1},
  pages={1--14},
  year={2001},
  publisher={Elsevier}
}

@article{sachidanandam2001map,
  title={A map of human genome sequence variation containing 1.42 million single nucleotide polymorphisms},
  author={Sachidanandam, Ravi and Weissman, David and Schmidt, Steven C and Kakol, Jerzy M and Stein, Lincoln D and Marth, Gabor and Sherry, Steve and Mullikin, James C and Mortimore, Beverley J and Willey, David L and others},
  journal={Nature},
  volume={409},
  number={6822},
  pages={928--934},
  year={2001},
  publisher={Nature Publishing Group}
}

@article{guerra2014note,
  title={A note on the simulation of overdispersed random variables with specified marginal means and product correlations},
  author={Guerra, Matthew W and Shults, Justine},
  journal={The American Statistician},
  volume={68},
  number={2},
  pages={104--107},
  year={2014},
  publisher={Taylor \& Francis}
}

@article{chaganty2006range,
  title={Range of correlation matrices for dependent Bernoulli random variables},
  author={Chaganty, N Rao and Joe, Harry},
  journal={Biometrika},
  volume={93},
  number={1},
  pages={197--206},
  year={2006},
  publisher={Oxford University Press}
}

@inbook{shults2014quasi,
  title={Quasi-least squares regression},
  author={Shults, Justine and Hilbe, Joseph M},
  chapter      = 7,
  pages        = {142--150},
  year={2014},
  publisher={CRC Press}
}

@article{prentice1988correlated,
  title={Correlated binary regression with covariates specific to each binary observation},
  author={Prentice, Ross L},
  journal={Biometrics},
  pages={1033--1048},
  year={1988},
  publisher={JSTOR}
}

@article{kennedy2003large,
  title={Large-scale genotyping of complex {DNA}},
  author={Kennedy, Giulia C and Matsuzaki, Hajime and Dong, Shoulian and Liu, Wei-min and Huang, Jing and Liu, Guoying and Su, Xing and Cao, Manqiu and Chen, Wenwei and Zhang, Jane and others},
  journal={Nature Biotechnology},
  volume={21},
  number={10},
  pages={1233},
  year={2003},
  publisher={Nature Publishing Group}
}

@book{hardin2002generalized,
  title={Generalized estimating equations},
  author={Hardin, James W and Hilbe, Joseph M},
  year={2002},
  publisher={Chapman and Hall/CRC}
}

@article{cox1972analysis,
  title={The analysis of multivariate binary data},
  author={Cox, David R},
  journal={Applied Statistics},
  pages={113--120},
  year={1972},
  publisher={JSTOR}
}

@article{carey1993modelling,
  title={Modelling multivariate binary data with alternating logistic regressions},
  author={Carey, Vincent and Zeger, Scott L and Diggle, Peter},
  journal={Biometrika},
  volume={80},
  number={3},
  pages={517--526},
  year={1993},
  publisher={Oxford University Press}
}

@article{preisser2014comparison,
  title={A comparison of methods for simulating correlated binary variables with specified marginal means and correlations},
  author={Preisser, John S and Qaqish, Bahjat F},
  journal={Journal of Statistical Computation and Simulation},
  volume={84},
  number={11},
  pages={2441--2452},
  year={2014},
  publisher={Taylor \& Francis}
}

@article{metzker2010sequencing,
  title={Sequencing technologies-the next generation},
  author={Metzker, Michael L},
  journal={Nature Reviews Genetics},
  volume={11},
  number={1},
  pages={31},
  year={2010},
  publisher={Nature Publishing Group}
}

@article{demirtas2006method,
  title={A method for multivariate ordinal data generation given marginal distributions and correlations},
  author={Demirtas, Hakan},
  journal={Journal of Statistical Computation and Simulation},
  volume={76},
  number={11},
  pages={1017--1025},
  year={2006},
  publisher={Taylor \&amp; Francis}
}

@article{naik2008challenges,
  title={Challenges and opportunities in high-dimensional choice data analyses},
  author={Naik, Prasad and Wedel, Michel and Bacon, Lynd and Bodapati, Anand and Bradlow, Eric and Kamakura, Wagner and Kreulen, Jeffrey and Lenk, Peter and Madigan, David M and Montgomery, Alan},
  journal={Marketing Letters},
  volume={19},
  number={3-4},
  pages={201},
  year={2008},
  publisher={Springer}
}

@article{fieuws2006high,
  title={High dimensional multivariate mixed models for binary questionnaire data},
  author={Fieuws, Steffen and Verbeke, Geert and Boen, Filip and Delecluse, Christophe},
  journal={Journal of the Royal Statistical Society: Series C (Applied Statistics)},
  volume={55},
  number={4},
  pages={449--460},
  year={2006},
  publisher={Wiley Online Library}
}

@article{wilbur2002variable,
  title={Variable selection in high-dimensional multivariate binary data with application to the analysis of microbial community {DNA} fingerprints},
  author={Wilbur, Jayson D and Ghosh, JK and Nakatsu, CH and Brouder, SM and Doerge, RW},
  journal={Biometrics},
  volume={58},
  number={2},
  pages={378--386},
  year={2002},
  publisher={Wiley Online Library}
}

@incollection{diwakar2009data,
  title={Data Quality for Decision Support--The Indian Banking Scenario},
  author={Diwakar, Hemalatha and Vaidya, Alka},
  booktitle={Data Quality and High-Dimensional Data Analysis},
  pages={60--77},
  year={2009},
  publisher={World Scientific}
}

@Manual{R,
  title = {\proglang{R}: {A} Language and Environment for Statistical Computing},
  author = {{\proglang{R} Core Team}},
  organization = {\proglang{R} Foundation for Statistical Computing},
  address = {Vienna, Austria},
  year = {2017},
  url = {https://www.R-project.org/},
}

@article{park1996simple,
  title={A simple method for generating correlated binary variates},
  author={Park, Chul Gyu and Park, Taesung and Shin, Dong Wan},
  journal={The American Statistician},
  volume={50},
  number={4},
  pages={306--310},
  year={1996},
  publisher={Taylor \&amp; Francis}
}

@article{emrich1991method,
  title={A method for generating high-dimensional multivariate binary variates},
  author={Emrich, Lawrence J and Piedmonte, Marion R},
  journal={The American Statistician},
  volume={45},
  number={4},
  pages={302--304},
  year={1991},
  publisher={Taylor \& Francis}
}

@article{leisch1998generation,
  title={On the generation of correlated artificial binary data},
  author={Leisch, Friedrich and Weingessel, Andreas and Hornik, Kurt},
  year={1998},
  journal={Working Papers SFB ``Adaptive Information Systems and Modelling in Economics and Management Science"},
  volume={13}, 
  publisher={WU Vienna University of Economics and Business},
  address={Vienna}
}

@TechReport{bahadur1959representation,
  title={A representation of the joint distribution of responses to n dichotomous items},
  author={Bahadur, Raj Raghu},
  year={1959},
  type={Technical {R}eport},
  institution={Columbia University New York Teachers College}
}

@article{gange1995generating,
  title={Generating multivariate categorical variates using the iterative proportional fitting algorithm},
  author={Gange, Stephen J},
  journal={The American Statistician},
  volume={49},
  number={2},
  pages={134--138},
  year={1995},
  publisher={Taylor \& Francis Group}
}

@article{lee1993generating,
  title={Generating random binary deviates having fixed marginal distributions and specified degrees of association},
  author={Lee, AJ},
  journal={The American Statistician},
  volume={47},
  number={3},
  pages={209--215},
  year={1993},
  publisher={Taylor \& Francis Group}
}

@article{lunn1998note,
  title={A note on generating correlated binary variables},
  author={Lunn, A Daniel and Davies, Stephen J},
  journal={Biometrika},
  volume={85},
  number={2},
  pages={487--490},
  year={1998},
  publisher={Oxford University Press}
}
\end{document}